\documentclass[pamq]{ipart}

\usepackage{amsmath, amsfonts}
\usepackage{amsthm}
\usepackage{float}
 \usepackage{amssymb}
 \usepackage{amssymb}
\usepackage{bbm}
\usepackage{wrapfig}
\usepackage{tikz}
\usepackage{hyperref,url}
\usepackage{thmtools}
\usepackage{thm-restate}
\usepackage{amssymb}
\usepackage{amscd}
\usepackage{amsthm}
\usepackage{enumerate}
\usepackage{graphicx}
\usepackage{mathtools}
\usepackage{tcolorbox} 
\usepackage{subfigure} 
\usepackage{siunitx}
\usepackage{tikz-cd}
\usepackage{bm,blkarray}
\numberwithin{equation}{section}

\startlocaldefs
\newcommand\eps{\varepsilon}
\DeclareMathOperator{\Vol}{Vol}

\DeclareMathOperator{\sgn}{sgn}

\newtheorem{theorem}{Theorem}[section]
\newtheorem{definition}[theorem]{Definition}
\newtheorem{lemma}[theorem]{Lemma}

\newtheorem{proposition}[theorem]{Proposition}
\newtheorem{question}[theorem]{Question}

\newtheorem*{remark}{Remark}

\endlocaldefs

\pubyear{2025}
\volume{21}
\issue{6}
\firstpage{2321}
\lastpage{2368}

\begin{document}

\begin{frontmatter}

\title[Distribution of points on the real line]{Distribution of points on the real line under a class of repulsive potentials}

\begin{aug}
\author{\fnms{Roni} \snm{Edwin}\ead[label=e1]{roni.edwin@yale.edu}}
\address{Department of Mathematics\\ Yale University\\ 219 Prospect Street\\ New Haven, CT\\ USA \\ \printead{e1}}
\end{aug}
\received{\sday{23} \smonth{3} \syear{2025}}

\begin{abstract}
In a 1979 paper, Ventevogel and Nijboer showed that classical point particles interacting via the pair potential $\phi(x)=\left(1+x^4\right)^{-1}$ are not equally spaced in their ground states in one dimension when the particle density is high, in contrast with many other potentials such as inverse power laws or Gaussians. In this paper, we explore a broad class of potentials for which this property holds; we prove that under the potentials $f_\alpha(x)=\left(1+x^\alpha\right)^{-1}$, when $\alpha>2$ is an even integer, there is a corresponding $s_\alpha>0$ such that under density $\rho={n}/{s_\alpha}$, the configuration that places $n$ particles at each point of $s_\alpha\mathbb{Z}$ minimises the average potential energy per particle and is therefore the exact ground state. In other words, the particles form clusters, while the clusters do not approach each other as the density increases; instead they maintain a fixed spacing. This is, to the best of our knowledge, the first rigorous analysis of such a ground state for a naturally occurring class of potential functions.
\end{abstract}



\tableofcontents

\end{frontmatter}


\section{Introduction} 
Potential energy minimisation arises out of trying to understand why systems of interacting particles often arrange themselves in crystalline structures at low temperatures \cite{radin1987low}. The classical model involves an infinite system of particles $C \subset \mathbb{R}^d$ interacting under a pair potential $f$, for example, an inverse power law like $f:x \mapsto x^{-\alpha}$, or a Gaussian $f:x \mapsto e^{-\beta x^2}$, and the goal is to arrange the particles so as to minimise the average potential energy per particle. In the case of repulsive potentials, we often constrain the particles by requiring they have a fixed density $\rho$, given by the limit \begin{align*}
    \rho=\lim_{r \to \infty}\frac{\#\left(C\cap B_r^d\right)}{\Vol\left(B_r^d\right)},
\end{align*} where $B_r^d$ is the ball of radius $r$ in $\mathbb{R}^d$, and $\Vol\left(\cdot \right)$ represents the $d$-dimensional Lebesgue measure. Given a potential $f:(0,\infty) \to \mathbb{R}$ or $f:[0,\infty) \to \mathbb{R}$,
for such a configuration $C=\left(x_i\right)_{i \in \mathbb{Z}}$, we define its lower $f$-energy, denoted $E_f(C)$, as the following limit inferior: \begin{equation}
    E_f(C)=\liminf_{r \to \infty}\frac1{\#\left(C \cap B_r^d\right)}\sum_{\substack{x_i,x_j \in C \cap B_r^d\\ i\neq j}}f\mathopen{}\left(\left|x_i-x_j\right|\right).
    \label{lowerfenergydefintiion}
\end{equation} This coincides with the idea of lower $f$-energy presented in Section $1$ of \cite{cohn2022universal}. Note the sum above is an ordered sum over the points. If the limit above exists, and not just the limit inferior, we simply refer to that as the $f$-energy of $C$. Perhaps the simplest case of symmetry and crystalline structure is when the configuration $C$ is a full-rank lattice $\Lambda$, that is the integer span of $d$ linearly independent vectors $v_1,v_2, . . . , v_d$ in $\mathbb{R}^d$. 
In this case, its density is the reciprocal of the volume of its \emph{fundamental cell}, $\mathbb{R}^d/\Lambda$, so \begin{equation*}
\rho=\frac1{\Vol\left(\mathbb{R}^d/\Lambda\right)},
\end{equation*} and its $f$-energy is given by \begin{equation}
    E_f(\Lambda)=\sum_{v \in \Lambda\setminus \{0\}}f\left(\left|v\right|\right),
    \label{energyoflatticegen}
\end{equation} assuming the sum above is absolutely convergent (equation 1.1 in \cite{cohn2022universal}). With this setup, the potential energy minimisation problem is as follows. Among such configurations of points $C$ with fixed density $\rho$, the goal is to then find a configuration $C^*$ of density $\rho$ that minimises lower $f$-energy. So $$E_f(C)\ge E_f\left(C^*\right)$$ for all $C$ with density $\rho$. We borrow terminology from \cite{cohn2022universal}, and refer to such a minimiser $C^*$ as a \emph{ground state} for $f$. The most interesting cases are when certain configurations are ground states for a large class of potentials. This leads us to the question of which potentials to consider. A natural choice might be to require convexity, because it means the strength of the repulsion gets weaker at larger distances; a more stringent condition would be to require complete monotonicity. Recall that a completely monotonic function $g:(0,\infty) \to \mathbb{R}$ is one that satisfies $(-1)^kg^{(k)}\ge 0$ for all $k\ge 0$. This leads us to the universal optimality theory described in \cite{cohn2007universally} and \cite{cohn2022universal}:
\begin{definition}
     Let $C$ be a point configuration in $\mathbb{R}^d$ with density $\rho > 0$. We say $C$ is \emph{universally optimal} if it minimises $f$-energy (among configurations of density $\rho$) whenever $f :(0,\infty) \to \mathbb{R}$ is a completely monotonic function of squared distance. So $f(r)=g\left(r^2\right)$, where $g$ is completely monotone.
\end{definition} In the case of $d=1$, Nijboer and Ventevogel proved in \cite{ventevogel1979configuration} that the equidistant configuration $\frac1{\rho}\mathbb{Z}$ is universally optimal for each $\rho>0$, with Cohn and Kumar \cite{cohn2007universally} giving a different proof. Its energy in this case is given by  \begin{equation}
    E_f\left(\frac1{\rho}\mathbb{Z}\right)=\sum_{n \in \mathbb{Z}\setminus\{0\}}f\left(\left|\frac{n}{\rho}\right|\right),
    \label{energyoflattice}
\end{equation}
 from \eqref{energyoflatticegen}. In higher dimensions, there is not a lot that is known about universal optimality save in some special cases. In dimensions $d=8$ and $d=24$, it has been proven in \cite{cohn2022universal} that the $E_8$ and Leech lattices respectively are universally optimal. Another interesting case is dimension $d=2$, where it is suspected that the $A_2$ root lattice, the hexagonal lattice, is universally optimal. It has been proven (see \cite{henn2016hexagonal} and \cite{montgomery1988minimal}) that the $A_2$ root lattice is universally optimal among lattices, but proving its universal optimality in general remains open. It might seem strange that we have positive results in dimensions $8$ and $24$, but none so far for $d=2$; we might expect the problem to increase in difficulty as the number of dimensions increases. It turns out however that the distances between the vectors in the $E_8$ and Leech lattices are nice enough to allow for some clever interpolation to prove they are indeed universally optimal (see \cite{radchenko2019fourier}, Theorem 1.7 in \cite{cohn2022universal}).
 
One interesting phenomenon that occurs in $\mathbb{R}$ absent in higher dimensions is the behavior of the ground states of repulsive ($f'<0$) convex potentials $f$. Nijboer and Ventevogel further proved in \cite{ventevogel1979configuration} that under decreasing convex potentials $f$, the configuration $\frac1{\rho}\mathbb{Z}$ is a ground state for each density $\rho>0$.  This seems to suggest the ground state property of the equidistant configuration $\frac1{\rho}\mathbb{Z}$ is more robust in $\mathbb{R}$. Given these results, it is natural to ask if there are other potentials for which the equidistant configuration is a ground state. One might hope for example, that the equidistant configuration $\frac1{\rho}\mathbb{Z}$ might minimise lower $f$-energy for repulsive potentials. It is not an entirely unreasonable assumption; the potential being repulsive means smaller energies when the points are further apart, so it kind of makes sense that the best way to arrange the points, fixing the density, while making them as far apart as possible would be to just make them equally spaced. It turns out however that this assumption is false. In particular the authors in \cite{nijboer1985minimum}, \cite{ventevogel1979configuration} showed that under the potential $\phi:x \mapsto \left(x^4+1\right)^{-1}$ with density $\rho=2$, $\frac1{2}\mathbb{Z}$ is sub-optimal by explicitly constructing a configuration of lower energy.

In this paper, we consider the behaviour of the ground states of the potential $x \mapsto \left(x^4+1\right)^{-1}$, and more broadly potentials of the form $f_\alpha\colon x \mapsto \left(x^\alpha+1\right)^{-1}$ when $\alpha>2$ is an even integer; these are natural potentials which can be thought of as a smooth approximation to a hard cut-off potential, and it is a problem which feels like we ought to be able to answer.  We formulate this in the following question:
\begin{question}
    Consider the potential $f_\alpha:x \mapsto \left(x^\alpha+1\right)^{-1}$ for $\alpha>2$ an even integer. What do the ground state configurations for $f_\alpha$ look like at high densities?
    \label{thequestion}
\end{question}
\begin{remark}
    The reason we consider $\alpha>2$ is that in general, for $\alpha\in [0,2]$, the function $x \mapsto \left(x^\alpha+1\right)^{-1}$ is a completely monotone function of squared distance, so the characterisation of its ground states is included as a special case of the work of Nijboer and Ventevogel in \cite{nijboer1985minimum}.
\end{remark}
It may be a bit ambitious to hope for a precise quantitative answer. We concern ourselves with high densities because this is where the interaction due to the potential comes up. Experimentally, the particles in ground state configurations seem to cluster together in roughly equally spaced clusters, shown by the following figures, for some values of $\alpha$ and density $\rho$:
\begin{figure}[H]
\minipage{0.5\textwidth}
\centering
  \begin{tikzpicture}[scale=0.18]
\draw (-15,0)--(15,0);
\foreach \x in {-15, -10, ..., 15}
\draw (\x,-0.5)--(\x,0.5);
\foreach \x in {-15, -10, ..., 15}
\draw (\x,-0.5) node[below] {$\x$};
\fill (4.099505310688963,0) circle (0.37416573867739417);
\fill (1.5451223313404543,0) circle (0.223606797749979);
\fill (7.752914507034196,0) circle (0.2645751311064591);
\fill (6.536773674742831,0) circle (0.2645751311064591);
\fill (-12.077133991020064,0) circle (0.37416573867739417);
\fill (14.895490208672552,0) circle (0.2645751311064591);
\fill (-3.5336056575206958,0) circle (0.28284271247461906);
\fill (0.3058475306542244,0) circle (0.316227766016838);
\fill (2.600997699242492,0) circle (0.2645751311064591);
\fill (10.23670190613861,0) circle (0.28284271247461906);
\fill (13.622452783930832,0) circle (0.28284271247461906);
\fill (-1.0509771174632279,0) circle (0.2645751311064591);
\fill (-7.9313173505650845,0) circle (0.34641016151377546);
\fill (-13.83877903464744,0) circle (0.37416573867739417);
\fill (-4.90680878355247,0) circle (0.30000000000000004);
\fill (-9.364900358564002,0) circle (0.2645751311064591);
\fill (11.562483908113304,0) circle (0.28284271247461906);
\fill (-10.580928186961982,0) circle (0.2645751311064591);
\fill (-6.362907632872521,0) circle (0.316227766016838);
\fill (8.966613468543523,0) circle (0.2645751311064591);
\fill (5.479376581022896,0) circle (0.223606797749979);
\fill (-2.2686507517116477,0) circle (0.2645751311064591);
\fill (12.592338816546254,0) circle (0.2);
\end{tikzpicture}
  \caption{Plot for $\alpha=4$, $\rho=8$}\label{1st}
\endminipage
\minipage{0.5\textwidth}%
 \centering
  \begin{tikzpicture}[scale=0.18]
\draw (-15,0)--(15,0);
\foreach \x in {-15, -10, ..., 15}
\draw (\x,-0.5)--(\x,0.5);
\foreach \x in {-15, -10, ..., 15}
\draw (\x,-0.5) node[below] {$\x$};
\fill (-0.9178685528389252,0) circle (0.458257569495584);
\fill (-4.670413886485573,0) circle (0.37416573867739417);
\fill (-5.937798970801928,0) circle (0.458257569495584);
\fill (-2.304072658694247,0) circle (0.43588989435406744);
\fill (3.0096611738231234,0) circle (0.4795831523312719);
\fill (-11.878257469951263,0) circle (0.4242640687119285);
\fill (7.909367144180079,0) circle (0.4242640687119285);
\fill (-14.644810369639409,0) circle (0.4242640687119285);
\fill (9.078879625868371,0) circle (0.36055512754639896);
\fill (5.598770254973633,0) circle (0.41231056256176607);
\fill (-13.261150943676155,0) circle (0.469041575982343);
\fill (-8.433733365713238,0) circle (0.3872983346207417);
\fill (6.740143108104943,0) circle (0.36055512754639896);
\fill (14.286312092220665,0) circle (0.41231056256176607);
\fill (0.44745306836660015,0) circle (0.4242640687119285);
\fill (10.220561486348371,0) circle (0.41231056256176607);
\fill (-7.256257131878139,0) circle (0.4);
\fill (-10.622435654611602,0) circle (0.4);
\fill (1.6777076021053154,0) circle (0.3872983346207417);
\fill (-9.51076882055521,0) circle (0.36055512754639896);
\fill (12.905294126469244,0) circle (0.4795831523312719);
\fill (-3.557139460330677,0) circle (0.3872983346207417);
\fill (11.502896125558618,0) circle (0.4242640687119285);
\fill (4.366402770092386,0) circle (0.4);
\end{tikzpicture}
  \caption{Plot for $\alpha=4$, $\rho=16$}\label{2nd}
\endminipage
\end{figure}
\begin{figure}[H]
\minipage{0.5\textwidth}
  \centering
  \begin{tikzpicture}[scale=0.18]
\draw (-15,0)--(15,0);
\foreach \x in {-15, -10, ..., 15}
\draw (\x,-0.5)--(\x,0.5);
\foreach \x in {-15, -10, ..., 15}
\draw (\x,-0.5) node[below] {$\x$};
\fill (-14.37052529346418,0) circle (0.316227766016838);
\fill (-12.895864477184345,0) circle (0.37416573867739417);
\fill (-11.397531948573523,0) circle (0.33166247903554);
\fill (-9.936798396267433,0) circle (0.34641016151377546);
\fill (-8.454757165135444,0) circle (0.34641016151377546);
\fill (-7.017101322156099,0) circle (0.316227766016838);
\fill (-5.559827720035753,0) circle (0.36055512754639896);
\fill (-4.189409755657393,0) circle (0.2645751311064591);
\fill (-2.7828930863445214,0) circle (0.3872983346207417);
\fill (-1.188854352336406,0) circle (0.3872983346207417);
\fill (0.42184187039960236,0) circle (0.4);
\fill (1.844202534778503,0) circle (0.2645751311064591);
\fill (3.122078885815801,0) circle (0.30000000000000004);
\fill (4.53441130934097,0) circle (0.34641016151377546);
\fill (5.973373953293665,0) circle (0.316227766016838);
\fill (7.238024442646638,0) circle (0.2449489742783178);
\fill (8.258294191466309,0) circle (0.223606797749979);
\fill (9.356542091689775,0) circle (0.2645751311064591);
\fill (10.661105873892977,0) circle (0.316227766016838);
\fill (12.055727664512288,0) circle (0.316227766016838);
\fill (13.495838647085925,0) circle (0.34641016151377546);
\fill (14.959480285987196,0) circle (0.33166247903554);
\end{tikzpicture}
  \caption{Plot for $\alpha=6$, $\rho=10$}\label{3rt}
\endminipage\hfill
\minipage{0.5\textwidth}%
  \centering
  \begin{tikzpicture}[scale=0.18]
\draw (-15,0)--(15,0);
\foreach \x in {-15, -10, ..., 15}
\draw (\x,-0.5)--(\x,0.5);
\foreach \x in {-15, -10, ..., 15}
\draw (\x,-0.5) node[below] {$\x$};
\fill (-14.774488114010138,0) circle (0.5196152422706632);
\fill (-13.403987530759222,0) circle (0.4242640687119285);
\fill (-12.09749012944639,0) circle (0.458257569495584);
\fill (-10.806351156378511,0) circle (0.41231056256176607);
\fill (-9.541457548369634,0) circle (0.43588989435406744);
\fill (-8.21014555806267,0) circle (0.469041575982343);
\fill (-6.841480982995437,0) circle (0.469041575982343);
\fill (-5.539355388625705,0) circle (0.41231056256176607);
\fill (-4.235150599694868,0) circle (0.469041575982343);
\fill (-2.868925256716342,0) circle (0.469041575982343);
\fill (-1.5493693486886846,0) circle (0.4242640687119285);
\fill (-0.27344008721257,0) circle (0.43588989435406744);
\fill (1.0223029160551558,0) circle (0.43588989435406744);
\fill (2.3831772207569983,0) circle (0.5);
\fill (3.7337865083563737,0) circle (0.4242640687119285);
\fill (5.024417027332933,0) circle (0.447213595499958);
\fill (6.146465284130804,0) circle (0.316227766016838);
\fill (7.2470398891706465,0) circle (0.43588989435406744);
\fill (8.55292468640791,0) circle (0.447213595499958);
\fill (9.927937594062069,0) circle (0.5);
\fill (11.337686541290669,0) circle (0.4795831523312719);
\fill (12.715538355957438,0) circle (0.469041575982343);
\fill (14.031373136691457,0) circle (0.4242640687119285);
\end{tikzpicture}
  \caption{Plot for $\alpha=6$, $\rho=20$}\label{4th}
\endminipage
\end{figure}
In the figures above the sizes of the dots correspond to the number of points in the cluster. As can be seen, the spacing of the clusters remains roughly the same, with the number of points in the clusters increasing as the density increases. Our results make this observed phenomenon precise; for example, we show that for the potential $x \mapsto \left(x^4+1\right)^{-1}$, the particles in the ground state configuration at high densities coalesce into clusters consisting of an equal number of points with a spacing of $\sqrt{2}$. More generally, when $\alpha>2$ is an even integer, the ground state at high densities for the potential $x \mapsto \left(x^\alpha+1\right)^{-1}$ coalesce into clusters consisting of an equal number points with spacing $s_\alpha>1$ (with $s_\alpha \to 1$ as $\alpha \to \infty$). This can be formulated as follows:
\begin{theorem}
   For each $\alpha>2$ that is an even integer, there is a unique $s_\alpha>0$ such that the following holds: Let $\rho={n}/{s_\alpha}$ with $n \in \mathbb{N}$ denote the density under consideration, and let $C \subset \mathbb{R}$ be a point configuration with density $\rho$. Let $\mathcal{C}_\rho$ be the configuration that places $n$ particles at each point of the lattice $s_\alpha\mathbb{Z}$. 
    Explicitly, $$\mathcal{C}_\rho=\bigcup_{r=-\infty}^\infty\left\{x_{r,s}\right\}_{s=1}^n \text{ where } x_{r,s}=s_\alpha r.$$ Then $\mathcal{C}_\rho$ is a ground state configuration of density $\rho$, under the potential $f_\alpha$.
    \label{minimise}
\end{theorem}
We expect this clustering behaviour for the potential $x\mapsto \left(x^\alpha+1\right)^{-1}$ to hold in general when $\alpha>2$, not just even integral values. We did attempt to prove this, though we felt the denseness and length of the proof did not justify the improvement on our results.

The paper is organised as follows: In Section \ref{piamia}, we introduce the ideas of discrete and continuous energy, and present a strengthened version of Theorem \ref{minimise} phrased in the language of continuous energy. In Section \ref{linprogbounds} we outline the structure of the proof of this strengthened theorem, which is carried out in Sections \ref{auxsec}, \ref{widehatpsi-alphapos}, and \ref{chamere}. 

We use an oscillatory factor of $-2\pi$ and a normalisation of $1$ for the Fourier transform. That is, for an integrable function $u \in L^1\mathopen{}\left(\mathbb{R}\right)\mathclose{}$, we define its Fourier transform $\widehat{u}:\mathbb{R} \to \mathbb{C}$ by \begin{equation}
      \widehat{u}(\xi)=\int_{\mathbb{R}}u(x)e^{-2\pi i x\xi}\textup{d}x.
  \end{equation} 

 \section{Discrete and continuous energy.}\label{piamia}
We briefly introduce the discrete energy problem as in Chapter $2$ of \cite{borodachov2019discrete}. The general setup described is as follows. We have a metric space $\left(A,\rho\right)$, (here $\rho$ is the metric) and a kernel $K:A\times A \to \mathbb{R} \cup \{+\infty\}$. Given an $N$-point configuration $\omega_N=\left\{x_i\right\}_{i=1}^N$ of points in $A$, we define their $K$-energy by \begin{equation}
     E_K\left(\omega_N\right)=\sum_{i=1}^N\sum_{\substack{j=1 \\ j\neq i}}^NK\mathopen{}\left(x_i,x_j\right)\mathclose{}.
 \end{equation} We will denote the infimum $K$-energy by \begin{equation}
     e_K\mathopen{}\left(A,N\right)\mathclose{}:=\inf\left\{E_K\left(\omega_N\right):\omega_N \subset A\right\}.
 \end{equation} In \cite{borodachov2019discrete}, the infimum energy is denoted by $\mathcal{E}_K(A,N)$, but we use the calligraphic $E$ to denote a different quantity, hence the lower case $e$ here. It is shown in \cite{borodachov2019discrete} that in the case where the kernel $K$ is lower semi-continuous and the set $A$ is compact, the infimum energy is in fact attained, so there is an \emph{optimal  $N$-point configuration}, denoted $\omega_N^*$, such that \begin{equation}
     E_K\left(\omega_N^*\right)=e_K(A,N).
     \label{optimalmeasure}
 \end{equation}
 We conclude this brief introduction to discrete energy with an important proposition about the asymptotic behaviour of the infimum energy as a function of $N$, for a fixed kernel $K$ and set $A$, which relates this idea of discrete energy to the idea of continuous energy to be introduced shortly:
 \begin{proposition}[Proposition 2.1.1 in \cite{borodachov2019discrete}] Let $A$ be an infinite set and $K :A\times A \to \mathbb{R}\cup \{+\infty\}$ be an arbitrary kernel. Then the sequence \begin{equation*}
     \left\{\frac{e_K(A,N)}{N(N-1)}\right\}_{N=2}^\infty
 \end{equation*} is non-decreasing with $N$.
 \label{increasingnergy}
 \end{proposition} A natural thing to ask is if the sequence described in the proposition above has some interpretation, for example if it converges to some limit. Here is where the idea of continuous energy comes in. The idea of continuous energy is similar to that of discrete energy introduced above, but instead of dealing with a finite number of points, we look at Borel measures. Here $A$ denotes a compact infinite set in $\mathbb{R}^d$, and $\mathcal{M}(A)$ the collection of all Borel probability measures supported on $A$. Let $K:A\times A \to (-\infty,\infty]$ be a symmetric, lower semi-continuous kernel. We then define the \emph{continuous $K$-energy of $\mu$}, for a measure $\mu \in \mathcal{M}(A)$, by \begin{equation}
    I_K[\mu]:=\int_A\int_{A}K(x,y)\textup{d}\mu(x)\textup{d}\mu(y).
    \label{continuousenergydefinition.}
\end{equation} The \emph{Wiener constant}, denoted $W_K(A)$ is the smallest such energy. That is, \begin{equation*}
    W_K(A):=\inf\left\{I_K[\mu]:\mu \in \mathcal{M}(A)\right\}.
\end{equation*} Under some mild conditions on the kernel $K$, the Wiener constant is actually achieved by some measure $\mu$. Specifically, if the set $A$ is infinite and compact, the kernel $K$ is symmetric and lower semi-continuous, Lemma 4.1.3 in \cite{borodachov2019discrete} guarantees the existence of an \emph{equilibrium measure} $\mu^*$, so \begin{equation*}
    I_K\left[\mu^*\right]=W_K(A).
\end{equation*}
Moreover, the sequence in Proposition \ref{increasingnergy} actually converges to $W_K(A)$. That is, \begin{equation*}
    \lim_{N \to \infty}\frac{e_K(A,N)}{N^2}=W_K(A).
\end{equation*}
Equivalently,  \begin{equation}
    \lim_{N \to \infty}\frac{E_K\left(\omega_N^*\right)}{N^2}=W_K(A),
    \label{relationdiscontfin}
\end{equation} where $\omega_N^*$ is an optimal $N$-point configuration, as in \eqref{optimalmeasure}.
This relation can be viewed as a sort of approximation. That is, for large values of $N$, the discrete energy problem is essentially `approximating' the infimum continuous $K$-energy, $W_K(A)$.

We hope to apply this relation between discrete and continuous energy in answering Question \ref{thequestion}: Lower $f$-energy is essentially a case of discrete energy, but for when the set $A=\mathbb{R}$, and we replace the number of points $N$ with the density $\rho$. We define a continuous analogue of lower $f$-energy, and answer the continuous energy problem in that case.

\subsection{Continuous energy on \texorpdfstring{$\mathbb{R}$.}{Lg}}
We are going to define continuous energy for Borel measures on the whole of $\mathbb{R}$, that generalises the idea of lower $f$-energy for point configurations. Let $\mathcal{U}$ be the set of Borel measures on $\mathbb{R}$ that have `average mass' $1$. Explicitly \begin{equation}
    \mathcal{U}=\left\{\mu: \mu \text{ is a Borel measure on } \mathbb{R} \text{ and } \lim_{r \to \infty}\frac{\mu\left([-r,r]\right)}{2r}=1 \right\}.
    \label{mathcalUsetdef}
\end{equation} For a given continuous potential $f:[0,\infty) \to \mathbb{R}$, we define the \emph{continuous $f$-energy} of a measure $\mu \in \mathcal{U}$, by \begin{equation}
    \mathcal{E}_f(\mu)=\liminf_{r \to \infty}\frac1{\mu\left([-r,r]\right)}\int_{[-r,r]}\int_{[-r,r]}f(|x-y|)\textup{d}\mu(x)\textup{d}\mu(y).
    \label{energyofmathcalEdef}
\end{equation} In the case where the function $f$ extends to an even function on $\mathbb{R}$, we will drop the absolute value in the argument of $f$. We denote by $e(f)$ the infimum of continuous $f$-energy on $\mathbb{R}$, so \begin{equation}
e(f)=\inf\left\{\mathcal{E}_f(\mu): \mu \in \mathcal{U}\right\}.
    \label{infimalinfinteconen}
\end{equation} If the infimum $f$-energy $e(f)$ is achieved by some measure $\mu^* \in \mathcal{U}$, so $\mathcal{E}_f\mathopen{}\left(\mu^*\right)\mathclose{}=e(f)$, we will call $\mu^*$ an \emph{equilibrium measure for $f$.}
To get the relation between lower $f$-energy and this idea of continuous energy, we introduce a similar idea to the normalised counting measure on a finite point configuration.
\begin{definition}
    For a general configuration $C$ in $\mathbb{R}$ of density $\rho>0$, let $\omega(C)$ be the \emph{normalised counting measure on $C$}, given by
  \begin{equation}
    \omega(C)=\frac1{\rho}\sum_{x \in C} \delta\mathopen{}\left(x\right)\mathclose{},
        \label{defomegaC2}
    \end{equation} where $\delta(x)$ is the Dirac measure at $x$. Importantly, $\omega(C) \in \mathcal{U}$.
    \label{defomegaC}
\end{definition}
This definition implies the following relation between $E_f(C)$ and $\mathcal{E}_f\left(\omega(C)\right)$:\begin{equation}
\mathcal{E}_f\left(\omega(C)\right)=\frac1{\rho}E_f\left(C\right)+\frac{f(0)}{\rho},
    \label{distocont}
\end{equation}
which follows from plugging in the definition of $\omega(C)$ into the formula above for $\mathcal{E}_f(\mu)$.
As an example, using the formula for $E_f(s\mathbb{Z})$ in \eqref{energyoflattice}, we get for $s>0$, \begin{equation}
\mathcal{E}_f\left(\omega(s\mathbb{Z})\right)=\sum_{n \in \mathbb{Z}} sf(sn).
    \label{energyoflatticeocntinuous}
\end{equation}
With these ideas laid down, we present the following strengthening of Theorem \ref{minimise}:
\begin{theorem}
Let $\alpha>2$ be an even integer. Then there is a unique $s_\alpha>0$ such that continuous $f_\alpha$-energy as defined in \eqref{energyofmathcalEdef}, $\mathcal{E}_{f_\alpha}$, is minimised by the measure\begin{equation}
\omega\left(s_\alpha\mathbb{Z}\right)=\sum_{n \in \mathbb{Z}} s_\alpha\delta\mathopen{}\left(s_\alpha n\right)\mathclose{},
\label{minimisingmeasuredef}
    \end{equation} where $\delta(x)$ is the Dirac measure at $x$. From \eqref{energyoflatticeocntinuous}, its energy is given by \begin{equation}
\mathcal{E}_{f_\alpha}\mathopen{}\left(\omega\mathopen{}\left(s_\alpha\mathbb{Z}\right)\mathclose{}\right)\mathclose{}=s_\alpha\sum_{n \in \mathbb{Z}}  f_\alpha\mathopen{}\left(s_\alpha n\right)\mathclose{}.
        \label{formenergy}
    \end{equation}
    \label{measure}
\end{theorem} 
Notice how Theorem \ref{measure} implies Theorem \ref{minimise}: For each $n \in \mathbb{N}$, let $\mathcal{C}_{\frac{n}{s_\alpha}}$ denote the point configuration that places $n$ particles at each point of the lattice $s_\alpha\mathbb{Z}$, so $\omega\left(\mathcal{C}_{ \frac{n}{s_\alpha}}\right)=\sum_{n \in \mathbb{Z}}s_\alpha\delta\mathopen{}\left(s_\alpha n\right)\mathclose{}$.
Then for any configuration $C$ of density $\frac{n}{s_\alpha}$, \begin{equation*}
    \mathcal{E}_{f_\alpha}\mathopen{}\left(\omega(C)\right)\mathclose{}\ge \mathcal{E}_{f_\alpha}\mathopen{}\left(\omega\left(\mathcal{C}_{\frac{n}{s_\alpha}}\right)\right)\mathclose{}.
\end{equation*} Applying the identity in \eqref{distocont} implies \begin{equation*}
     \frac1{\rho}E_{f_\alpha}(C)+\frac{f_\alpha(0)}{\rho}\ge \frac1{\rho}E_{f_\alpha}\mathopen{}\left(\mathcal{C}_{\frac{n}{s_\alpha}}\right)\mathclose{}+\frac{f_\alpha(0)}{\rho},
\end{equation*}  and rearranging gives $E_{f_\alpha}(C)\ge E_{f_\alpha}\mathopen{}\left(\mathcal{C}_{\frac{n}{s_\alpha}}\right)$, which is the statement of Theorem \ref{minimise}.\newline

       \noindent
       \newline

       \noindent
\section{Structure of the proof}\label{linprogbounds}
Here we outline how the proof of Theorem \ref{measure} is presented.  The first thing to note is that from the formula for $\mathcal{E}_f\mathopen{}\left(\omega(t\mathbb{Z})\right)\mathclose{}$ in \eqref{energyoflatticeocntinuous}, $s_\alpha$ necessarily minimises \begin{equation*}
t\mapsto \sum_{n \in \mathbb{Z}} tf_\alpha(tn)
\end{equation*} over $t>0$. For $\alpha>2$, it turns out that this function does have a unique minimum $s_\alpha$ with $s_\alpha>1$. To that end, we rescale the potential $f_\alpha$, defining $F_\alpha \colon \mathbb{R} \to [0,\infty)$ by\begin{equation}
F_\alpha(x)=\frac1{1+s_\alpha^\alpha x^\alpha},
    \label{capFalpha}
\end{equation} so that to prove Theorem \ref{measure}, it is equivalent to show the counting measure on $\mathbb{Z}$, $\sum_{n \in \mathbb{Z}}\delta\mathopen{}\left(n\right)\mathclose{}$, minimises continuous $F_\alpha$-energy. We now present an analogous proposition to the linear programming bounds in \cite{cohn2022universal}, but for measures instead of discrete point configurations. We start by presenting the statement of the linear programming bounds:
\begin{proposition}[Proposition 2.2 in \cite{gausscoremodel}]
Let $f:(0,\infty) \to [0,\infty)$ be any function, and suppose
$h: \mathbb{R}^n \to \mathbb{R}$ is continuous, integrable, and positive definite (so $\widehat{h}\ge 0$). If $h(x)\le  h\left(|x|\right)$ for all $x \in \mathbb{R}^d\setminus\{0\}$, then every subset of $\mathbb{R}^n$ with density $\rho$ has 
lower $f$-energy at least $\rho\widehat{h}(0)-h(0)$.
\end{proposition}
\noindent
 Writing the inequality from the proposition as \begin{align*}
    \frac1{\rho}E_f(C)+\frac{h(0)}{\rho}\ge \widehat{h}(0)
\end{align*} motivates the following: 
\begin{proposition}
    Let $f:\mathbb{R} \to \mathbb{R}$ be an even continuous integrable function, with non-negative Fourier transform $\widehat{f}$. Then for each Borel measure $\mu$ with average mass $1$, so $\mu \in \mathcal{U}$ where $\mathcal{U}$ is as defined in \eqref{mathcalUsetdef}, $\mathcal{E}_f(\mu)\ge \widehat{f}(0)$.
    \label{lp}
\end{proposition}
Observe this proposition is equivalent to the linear programming bounds (with $h=f$) if $\mu$ is the normalised counting measure on a multiset $C \subset \mathbb{R}$ with a well-defined density; to prove this proposition, we note the proof of the linear programming bounds given in Proposition 2.2 in \cite{gausscoremodel} works for a general measure $\mu \in \mathcal{U}$, not just when it is the normalised counting measure on a multiset. 

Once we have Proposition \ref{lp}, the next step is to construct a suitable auxiliary function $\psi_\alpha$ such that \begin{align}
&\psi_\alpha\le F_\alpha, \quad\text{ and }\label{ghost} \\
    &\mathcal{E}_{F_\alpha}\mathopen{}\left(\sum_{n \in \mathbb{Z}}\delta(n)\right)=\mathcal{E}_{\psi_\alpha}\mathopen{}\left(\sum_{n \in \mathbb{Z}}\delta\mathopen{}\left( n\right)\mathclose{}\right)=\widehat{\psi_\alpha}(0).
    \label{likehim}
\end{align} In this case, this would imply the measure $\sum_{n \in \mathbb{Z}}\delta\mathopen{}\left( n\right)\mathclose{}$ minimises continuous $F_\alpha$-energy: For any measure $\mu \in \mathcal{U}$, we have $\mathcal{E}_{\psi_\alpha}(\mu)\le \mathcal{E}_{F_\alpha}(\mu)$ since $\psi_\alpha\le F_\alpha$, and since $\widehat{\psi_\alpha}\ge 0$, this combined with Proposition \ref{lp} implies $\mathcal{E}_{F_\alpha}(\mu)\ge \widehat{\psi_\alpha}(0)$. Combining this with \eqref{likehim} tells us \begin{align*}
    \mathcal{E}_{F_\alpha}(\mu)\ge  \mathcal{E}_{\psi_\alpha}\mathopen{}\left(\sum_{n \in \mathbb{Z}}\delta\mathopen{}\left( n\right)\mathclose{}\right)=\mathcal{E}_{F_\alpha}\mathopen{}\left(\sum_{n \in \mathbb{Z}}\delta\mathopen{}\left( n\right)\mathclose{}\right)
\end{align*} for any $\mu \in \mathcal{U}$, which shows 
 the measure $\sum_{n \in \mathbb{Z}}\delta\mathopen{}\left( n\right)\mathclose{}$ minimises continuous $F_\alpha$-energy.
 
 So now our goal is to construct a continuous, integrable, positive-definite function $\psi_\alpha$ such that \eqref{ghost} and \eqref{likehim} hold. To do this,  we claim it is sufficient to construct $\psi_\alpha$ satisfying the following hypothesis:
\begin{restatable}{hypothesis}{hypothesisonpsialpha}
    Let $\psi_\alpha:\mathbb{R} \to \mathbb{R}$ be an even, continuous, integrable function satisfying the following five conditions: \begin{align}
    \psi_\alpha(x)& \le F_\alpha(x) \quad \text{for all }  x \in \mathbb{R}, \label{firstcondpsiog} \\
    \psi_\alpha(n)&=F_\alpha(n) \quad \text{for all }n\in \mathbb{N}, \label{seconcondpsiog} \\
    \widehat{\psi_\alpha} (\xi)&\ge 0 \quad \text{for all } \xi \in  \mathbb{R},  \label{thirdcondpsihatog} \\
    \widehat{\psi_\alpha}(\xi)&=0 \quad \text{if } |\xi|\ge 1, \label{fourthcondpsihatog}
\end{align} and \begin{align}
        \left|\psi_\alpha(x)\right|\le \frac{C}{1+|x|^2}\quad\text{ for all }x\in\mathbb{R},
        \label{sev}
    \end{align} for some constant $C>0$. Then the counting measure on $\mathbb{Z}$, $\sum_{n \in \mathbb{Z}}\delta(n)$, minimises continuous $F_\alpha$-energy, where $F_\alpha$ is as in \eqref{capFalpha}. This is equivalent to Theorem \ref{measure}.
    \label{thm:hypothesisonpsialpha}
\end{restatable}
\begin{proof}[Proof]
\eqref{thirdcondpsihatog} means $\psi_\alpha$ is positive-definite, and \eqref{firstcondpsiog} is saying $\psi_\alpha\le F_\alpha$, so what remains is to show \eqref{likehim}. Note \eqref{seconcondpsiog} implies \begin{align}
    \mathcal{E}_{F_\alpha}\mathopen{}\left(\sum_{n \in \mathbb{Z}}\delta(n)\right)\mathclose{}=\mathcal{E}_{\psi_\alpha}\mathopen{}\left(\sum_{n \in \mathbb{Z}}\delta(n)\right)\mathclose{}.
    \label{juliana}
\end{align} To show the common value of both sides above is $\widehat{\psi_\alpha}$, we apply the Poisson Summation Formula (PSF). As phrased in Lemma 1.11.5 in \cite{borodachov2019discrete}, it says that if $f:\mathbb{R} \to \mathbb{C}$ is a continuous function satisfying \begin{align}
        \left|f(x)\right|+\left|\widehat{f}(x)\right|\le \frac{C}{\left(1+|x|\right)^{1+\eps}}
        \label{hypothesispsf}
    \end{align} for some $\eps>0$, $C>0$, then  $\sum_{n \in \mathbb{Z}} f(n)=\sum_{n \in \mathbb{Z}} \widehat{f}(n)$. \eqref{fourthcondpsihatog} and \eqref{sev} combined imply  \begin{align*}
        \left|\widehat{\psi_\alpha}(x)\right|+\left|\psi_\alpha(x)\right|\le \frac{C'}{1+|x|^2},
    \end{align*} so we can apply PSF to get 
\begin{align*}
        \mathcal{E}_{\psi_\alpha}\mathopen{}\left(\sum_{n \in \mathbb{Z}}\delta(n)\right)\mathclose{}=\sum_{n \in \mathbb{Z}} \psi_\alpha(n)=\sum_{n \in \mathbb{Z}} \widehat{\psi_\alpha}(n)=\widehat{\psi_\alpha}(0),
    \end{align*} since from \eqref{fourthcondpsihatog}, $\psi_\alpha$ vanishes outside of $[-1,1]$. This equality combined with \eqref{juliana} implies \eqref{likehim}, as desired.
    \end{proof}
It turns out that conditions \ref{firstcondpsiog}, \eqref{seconcondpsiog}, and \eqref{fourthcondpsihatog} can be used to uniquely determine $\psi_\alpha$, and we construct $\psi_\alpha$ in Section \ref{auxsec}. Based on this construction, we are then able to see that $\psi_\alpha$ satisfies \eqref{seconcondpsiog}, \eqref{fourthcondpsihatog} and \eqref{sev}. We then prove \eqref{firstcondpsiog} in Section \ref{chamere}, and prove \eqref{thirdcondpsihatog} in Section \ref{widehatpsi-alphapos}, which would complete the proof that $\psi_\alpha$ satisfies Hypothesis \ref{thm:hypothesisonpsialpha},  and consequently prove Theorem \ref{measure}.
\begin{remark}[A word on notation]
      Often we would want to show an inequality like $f(x)\le g(x)$, and we do this by showing something like $f(x)\le F(x)$, $g(x)\ge G(x)$ and showing $F(x)\le G(x)$. To differentiate between already established inequalities like $F(x)\le G(x)$, and inequalities we want to show, like $f(x)\le g(x)$, we will use the notation $f(x)\underset{\textup{wts}}{\le} g(x)$ to indicate this is an inequality we want to show (wts).
\end{remark}
Before constructing the auxiliary functions, we first obtain useful estimates on $s_\alpha$ and $F_\alpha$, which we use in proving the desired inequalities for large values of $\alpha$.

\section{Estimates on \texorpdfstring{$s_\alpha$}{Lg} and \texorpdfstring{$F_\alpha$}{Lg}}
Observe that if $\sum_{n \in \mathbb{Z}}\delta(n)$ minimises continuous $F_\alpha$-energy, then from the formula for continuous energy in \eqref{energyoflatticeocntinuous},  \begin{equation}
    \sum_{n \in \mathbb{Z}} tF_\alpha\left(tn\right)\ge \sum_{n \in \mathbb{Z}} F_\alpha(n) \label{lomo}
\end{equation} for all $t>0$, and by differentiating the left-hand side of \eqref{lomo} and setting it equal to $0$ at $t=1$,  we get \begin{align}
    \sum_{n \in \mathbb{Z}} \left(F_\alpha(n)+nF_\alpha'(n)\right)=0. \label{accfirstordercond}
\end{align} This identity will play a rule in proving the inequalities for $\psi_\alpha$ listed in Hypothesis \ref{thm:hypothesisonpsialpha}. The first estimate is as follows.
\begin{lemma}
    Let $\beta>1$, and $k> 0$. Then \begin{align*}
        \sum_{\substack{n \in \mathbb{N}, n\ge k}}^\infty \frac1{n^\beta}\le \frac1{k^\beta}\left(\frac{\beta+k-1}{\beta-1}\right).
    \end{align*}
    \label{estimateforsumoverpowers}
\end{lemma}
This follows from upper bounding the Riemann sum on the left-hand side by the corresponding integral.
        We now will list some propositions characterising $s_\alpha$. The first shows that  $s_4=\sqrt{2}$:
\begin{restatable}{proposition}{s4issqrt2prop}
    $s_4=\sqrt{2}$, that is, the function $t \mapsto \sum_{n \in \mathbb{Z}} t\left(1+t^4n^4\right)^{-1}$ for $t>0$ is minimised at $t=\sqrt{2}$.
    \label{thm:s4issqrt2prop}
\end{restatable}
    \begin{proof}[Proof of Proposition \ref{thm:s4issqrt2prop}]
    From equation $9$ in \cite{nijboer1985minimum}, we have \begin{align*}
         \sum_{n=1}^\infty f_4(tn)=\sum_{n=1}^\infty \frac1{1+(tn)^4}=\frac{\pi \sqrt{2}}{4t}\cdot \frac{\sinh\mathopen{}\left(\frac{\pi \sqrt{2}}{t}\right)\mathclose{}+\sin \mathopen{}\left(\frac{\pi \sqrt{2}}{t}\right)\mathclose{}}{\cosh\mathopen{}\left(\frac{\pi \sqrt{2}}{t}\right)\mathclose{}-\cos\mathopen{}\left(\frac{\pi \sqrt{2}}{t}\right)\mathclose{}}-\frac1{2},
     \end{align*} so \begin{equation}
        \begin{split}
             \mathcal{E}_{f_4}(\omega(t\mathbb{Z}))&=\sum_{n=-\infty}^\infty tf_4(tn)=t+2t\sum_{n=1}^\infty f_4(tn) \\
             &=\frac{\pi}{\sqrt{2}}\cdot \frac{\sinh\mathopen{}\left(\frac{\pi \sqrt{2}}{t}\right)\mathclose{}+\sin \mathopen{}\left(\frac{\pi \sqrt{2}}{t}\right)\mathclose{}}{\cosh\mathopen{}\left(\frac{\pi \sqrt{2}}{t}\right)\mathclose{}-\cos\mathopen{}\left(\frac{\pi \sqrt{2}}{t}\right)\mathclose{}}.
        \end{split}
         \label{party}
     \end{equation} Differentiating the right-hand side with respect to $x=\frac{\pi\sqrt{2}}{t}$ and setting it equal to $0$ implies \begin{align*}
         \left(\cosh x+\cos x\right)\left(\cosh\left(x\right)-\cos x\right)-\left(\sinh\left(x\right)+\sin x\right)^{2}=0,
     \end{align*} which can be simplified to obtain $(\sinh x)( \sin x)=0$. This implies $x=\pi m$ for $m \in \mathbb{N}$, so $t=\frac{\sqrt{2}}{m}$. Plugging this into \eqref{party} and simplifying, we get \begin{align*}
         \mathcal{E}_{f_4}\mathopen{}\left(\frac{\sqrt{2}}{m}\mathbb{Z}\right)\mathclose{}&=1+\frac{2\left(-1\right)^{m}}{e^{\pi m}-\left(-1\right)^{m}},
     \end{align*} which we can see is minimised when $m=1$, so $t=\sqrt{2}$.  
     This shows $s_4=\sqrt{2}$ as desired.
\end{proof}

The next Proposition states that $s_\alpha>1$ for $\alpha\ge 6$:
\begin{restatable}{proposition}{minimiserprop}
    Let $\alpha\ge 6$ be an even integer, and let $s_\alpha$ minimise $t \mapsto \mathcal{E}_{f_\alpha}(\omega(t\mathbb{Z}))$. Then $s_\alpha> 1$.
    \label{thm:minimiserprop}
\end{restatable}
This can be verified with interval arithmetic \cite{IntervalArithmetic.jl} in Julia \cite{bezanson2017julia} when $6\le \alpha\le 14$, carried out in the following Github repository \cite{edwin2024}. We prove it for $\alpha\ge 16$. 
\begin{proof}[Proof of Proposition \ref{thm:minimiserprop} for $\alpha\ge 16$]
     Recall $\mathcal{E}_{f_\alpha}(\omega(t\mathbb{Z}))$ is given by \begin{equation*}
    \mathcal{E}_{f_\alpha}\mathopen{}\left(\omega(t\mathbb{Z})\right)\mathclose{}=\sum_{n\in \mathbb{Z}} tf_\alpha(tn).
\end{equation*} Suppose for the sake of contradiction that $s_\alpha\le 1$. We get a lower bound for $\mathcal{E}_f\mathopen{}\left(\omega(s_\alpha \mathbb{Z})\right)\mathclose{}$ in this case. If $\frac{3}{4}\le s_\alpha\le 1$, then $f_\alpha\mathopen{}\left(s_\alpha\right)\mathclose{}\ge \frac1{2}$ since $f_\alpha(x)$ is decreasing in $|x|$, so 
     \begin{align*}
\mathcal{E}_{f_\alpha}\mathopen{}\left(\omega(s_\alpha \mathbb{Z})\right)\mathclose{}&=\sum_{n\in\mathbb{Z}} s_\alpha f\mathopen{}\left(ns_\alpha\right)\mathclose{}\ge \sum_{n=-1}^1 s_\alpha f\mathopen{}\left(ns_\alpha\right)\mathclose{}=s_\alpha +2s_\alpha f\left(s_\alpha\right) \\
    &\ge \frac3{4}+2\cdot \frac3{4}\cdot \frac1{2}=\frac3{2}.
\end{align*}
If $\frac1{2}\le s_\alpha\le \frac3{4}$, then \begin{align*}
    \mathcal{E}_{f_\alpha}\mathopen{}\left(\omega(s_\alpha\mathbb{Z})\right)\mathclose{}&\ge \sum_{n=-1}^1 s_\alpha f\mathopen{}\left(ns_\alpha\right)\mathclose{} =s_\alpha+ 2s_\alpha f\left(s_\alpha\right)\ge \frac1{2}+2\cdot \frac1{2}\cdot f_\alpha\left(\frac3{4}\right) \\
    &=\frac1{2}+\frac1{1+\left(\frac{3}{4}\right)^\alpha},
\end{align*} so putting these two results together, if $\frac1{2}<s_\alpha\le 1$, then \begin{equation}
\mathcal{E}_{f_\alpha}\mathopen{}\left(\omega(s_\alpha\mathbb{Z})\right)\mathclose{}\ge\frac1{2}+\frac1{1+\left(\frac{3}{4}\right)^\alpha}.
    \label{lowboundefalpha}
\end{equation}
In general, if $\frac{1}{k}<s_\alpha\le \frac{1}{k-1}$ for $k\ge 3$ an integer, then since $f_\alpha(x)$ is decreasing in $|x|$, \begin{align*}
    \mathcal{E}_{f_\alpha}\mathopen{}\left(\omega(s_\alpha \mathbb{Z})\right)\mathclose{}&\ge s_\alpha +2s_\alpha \sum_{n=1}^{k-1}f_\alpha\mathopen{}\left(ns_\alpha\right)\mathclose{}\ge \frac1{k}+\frac2{k}\sum_{n=1}^{k-1}f_\alpha\mathopen{}\left(\frac{n}{k-1}\right)\mathclose{}.
\end{align*} If $k=3$, this reduces to \begin{align}
    \mathcal{E}_{f_\alpha}\mathopen{}\left(\omega(s_\alpha\mathbb{Z})\right)\mathclose{}\ge \frac2{3}+\frac2{3}f_\alpha\mathopen{}\left(\frac1{2}\right), \quad \text{ if }\frac1{3}<s_\alpha\le \frac1{2}.
    \label{monsoon}
\end{align} If $k\ge 4$, we have
\begin{align*}
    \mathcal{E}_{f_\alpha}\mathopen{}\left(\omega(s_\alpha \mathbb{Z})\right)\mathclose{}&\ge \frac1{k}+\frac2{k}\sum_{n=1}^{p}f_\alpha\mathopen{}\left(\frac{n}{k-1}\right)\mathclose{}+\frac2{k}\sum_{n=p+1}^{k-1}f_\alpha\mathopen{}\left(\frac{n}{k-1}\right)\mathclose{} \\
    &\ge \frac1{k}+\frac{2p}{k}f_\alpha\mathopen{}\left(\frac{p}{k-1}\right)\mathclose{}+\frac{k-p-1}{k},
\end{align*} so \begin{align}
    \mathcal{E}_{f_\alpha}\mathopen{}\left(\omega(s_\alpha \mathbb{Z})\right)\mathclose{}\ge \frac{k-p}{k}+\frac{2p}{k}f_\alpha\left(\frac{p}{k-1}\right)
    \label{tra}
\end{align}
for any $1\le p\le k-1$. We then choose $p=\left\lceil \frac{k}{2}\right\rceil$,$\left\lfloor \frac{k}{2}\right\rfloor$ in \eqref{tra} and average the two resulting inequalities to get \begin{align*}
    \mathcal{E}_{f_\alpha}\left(\omega(s_\alpha \mathbb{Z})\right)\ge \frac1{2}+f_\alpha\mathopen{}\left(\frac{\left\lceil \frac{k}{2}\right\rceil}{k-1}\right)\mathclose{}\ge \frac1{2}+\frac1{1+\left(\frac3{4}\right)^\alpha}\ge \frac2{3}\left(1+\frac1{1+2^{-\alpha}}\right),
\end{align*} if $k\ge 4$ and $\alpha\ge 16$. Combining this with  \eqref{lowboundefalpha},  we see that in general,
\begin{equation}
    \mathcal{E}_{f_\alpha}\mathopen{}\left(\omega(s_\alpha\mathbb{Z})\right)\mathclose{}\ge \frac2{3}\left(1+\frac1{1+2^{-\alpha}}\right)\quad\text{ if  }s_\alpha<1.
    \label{bound2}
\end{equation}
Now let $\widehat{s}=(2\alpha)^{\frac1{\alpha}}$.
Applying Lemma \ref{estimateforsumoverpowers},  \begin{align*}
   \mathcal{E}_{f_\alpha}\mathopen{}\left(\omega\mathopen{}\left(\widehat{s}\mathbb{Z}\right)\mathclose{}\right)\mathclose{}&=\sum_{n\in \mathbb{Z}}\frac{\widehat{s}}{1+n^\alpha \widehat{s}^\alpha}\le \widehat{s}+2\widehat{s}^{1-\alpha}+2\widehat{s}^{1-\alpha}\sum_{n=2}^\infty \frac1{n^\alpha} \\
   &\le \widehat{s}\left(1+\frac2{\widehat{s}^\alpha}+\frac2{2^\alpha \widehat{s}^\alpha}\left(\frac{\alpha+1}{\alpha-1}\right)\right) \\
   &=\left(2\alpha\right)^\frac1{\alpha}\mathopen{}\left(1+\frac1{\alpha}+\frac1{\alpha 2^\alpha}\left(\frac{\alpha+1}{\alpha-1}\right)\right)\mathclose{},
\end{align*} and combining this with \eqref{bound2}, we get \begin{equation*}
   \frac2{3}\left(1+\frac1{1+2^{-\alpha}}\right)\le \left(2\alpha\right)^{\frac1{\alpha}}\left(\frac{\alpha+1}{\alpha}\right)\left(1+\frac1{2^\alpha(\alpha-1)}\right),
\end{equation*} since $\mathcal{E}_{f_\alpha}\mathopen{}\left(\omega(s_\alpha \mathbb{Z})\right)\mathclose{}\le \mathcal{E}_{f_\alpha}\mathopen{}\left(\omega(s\mathbb{Z})\right)\mathclose{}$ for any $s>0$. The above inequality fails if $\alpha\ge 16$, so then we see that if $\alpha\ge 16$, then $s_\alpha>1$, as desired.
\end{proof}
Recall we defined $F_\alpha(x)$ by $F_\alpha(x)=\left(1+s_\alpha^\alpha x^\alpha\right)^{-1}$. From this we get the following expressions for $F_\alpha'$, $F_\alpha''$ in terms of $F_\alpha$:
\begin{align*}
    F_\alpha'(x)=-\frac{\alpha s_\alpha^\alpha x^{\alpha-1} }{\left(1+s_\alpha^\alpha x^\alpha\right)^2}=-\frac{\alpha }{x}\cdot \frac{s_\alpha ^\alpha x^\alpha+1-1}{\left(1+s_\alpha ^\alpha x^\alpha\right)^2}=-\frac{\alpha}{x}\left(F_\alpha(x)-F_\alpha(x)^2\right),
\end{align*}
so \begin{align}
    F_\alpha'(x)=-\frac{\alpha F_\alpha(x)\left(1-F_\alpha (x)\right)}{x}.
    \label{emu}
\end{align} Differentiating this yields \begin{align*}
    F_\alpha''(x)=-\frac{\alpha F_\alpha'(x)\left(1-F_\alpha(x)\right)}{x}+\frac{\alpha F_\alpha (x)}{x}\left(\frac{1-F_\alpha(x)}{x}+F_\alpha'(x)\right),
\end{align*}
and 
  substituting the expression for $F_\alpha'(x)$ in \eqref{emu}, and simplifying, we get \begin{align*}
  F_\alpha''(x)&=\alpha F_\alpha(x)\left(\alpha\left(1-F_\alpha(x)\right)+1-\alpha F_\alpha(x)\right)\left(\frac{1-F_\alpha(x)}{x^2}\right),
  \end{align*} or \begin{align}
      F_\alpha''(x)&=\alpha F_\alpha(x)\left(\frac{1-F_\alpha(x)}{x^2}\right)\left(\alpha\left(1-2F_\alpha(x)\right)+1\right).
      \label{Falphaduobleprime}
  \end{align}
Our third proposition gives an asymptotic expression for $s_\alpha$:
\begin{restatable}{proposition}{asympsalpha}
    For $\alpha\ge 12$, we can write \begin{equation}
    s_\alpha^\alpha=\alpha-2+\sqrt{\left(\alpha-2\right)^2-3}+\mathfrak{G}(\alpha),
    \label{ssubalpjha^alphaeq}
\end{equation} where \begin{equation}
    \left|\mathfrak{G}(\alpha)\right|\le \frac{2\alpha}{0.99^2}\cdot \frac{\left(\alpha+1\right)^2}{2^{\alpha-1}\left(\alpha-1\right)}.
    \label{jermaine} 
\end{equation} This bound implies $2\alpha-5 \le s_\alpha^\alpha\le 2\alpha-3$,
 \begin{align}
\frac1{2\alpha-2}\le F_\alpha(1)\le \frac1{2\alpha-4},
    \label{boundonFalphaof1}
\end{align}
and \begin{align}
   \frac{1}{2}\left(1-\frac1{2\alpha-4}\right)\le -F_\alpha'(1)\le \frac{\alpha}{2\alpha-4}.
   \label{boundFalpha'}
\end{align}
    \label{thm:asympsalpha}
\end{restatable}
\begin{proof}
    Using Proposition \ref{thm:minimiserprop}, for $\alpha\ge 12$, we are going to choose $t >1$ to minimise  \begin{align*}
E_\alpha(t):=\mathcal{E}_{f_\alpha}\mathopen{}\left(t\sum_{n \in \mathbb{Z}}\delta(tn)\right)\mathclose{}.
    \end{align*} Let \begin{align}
    \mathfrak{F}_\alpha(t)&=t+2tf_\alpha(t), \label{mathfrakfdef} \\
    \mathfrak{E}_\alpha(t)&=2\sum_{n= 2}^\infty tf_\alpha(nt), \label{mathfrakede}
\end{align} so that $E_\alpha(t)=\mathfrak{F}_\alpha(t)+\mathfrak{E}_\alpha(t)$. Since $s_\alpha$ minimises $E_\alpha(t)$, $E_\alpha'\mathopen{}\left(s_\alpha\right)\mathclose{}=0$, or
 \begin{equation}
\mathfrak{F}_\alpha'\mathopen{}\left(s_\alpha\right)\mathclose{}=-\mathfrak{E}_\alpha'\mathopen{}\left(s_\alpha\right)\mathclose{}.
    \label{condonsalpha}
\end{equation} We start by bounding $\mathfrak{E}_\alpha'(t)$ for $t>1$. From \eqref{emu} (that formula for the derivative is independent of $s_\alpha$), we may deduce \begin{equation}
\begin{split}
     \left|\mathfrak{E}_\alpha'\left(t\right)\right|&\le 2\sum_{n=2}^\infty \left|f_\alpha(nt)+tnf_\alpha'\left(nt\right)\right|\le 2\sum_{n=2}^\infty (\alpha+1)f_\alpha(nt) \\
     &\le \sum_{n=2}^\infty \frac{2\alpha+2}{n^\alpha}\le  \frac{\left(\alpha+1\right)^2}{2^{\alpha-1}\left(\alpha-1\right)},
\end{split}
   \label{boundonmathfrakEderivative}
\end{equation}  from Lemma \ref{estimateforsumoverpowers}.
We now differentiate $\mathfrak{F}_\alpha(t)$. From \eqref{mathfrakfdef}, \begin{equation*}
    \mathfrak{F}_\alpha'(t)=1+2\partial_t\mathopen{}\left(\frac{t}{1+t^\alpha}\right)\mathclose{}=1+2\left(\frac{1+t^\alpha-\alpha t^{\alpha}}{\left(1+t^\alpha\right)^2}\right)=1+2\left(\frac{1-\left(\alpha-1\right)t^\alpha}{\left(1+t^\alpha\right)^2}\right).
\end{equation*} Setting this equal to $-\mathfrak{E}_\alpha'(t)$ and solving, we get \begin{equation*}
    2\left(\frac{1-\left(\alpha-1\right)t^\alpha}{\left(1+t^\alpha\right)^2}\right)=-1-\mathfrak{E}_\alpha'(t),
\end{equation*} so \begin{equation*}
    2-2\left(\alpha-1\right)t^\alpha=-\left(1+\mathfrak{E}_\alpha'(t)\right)\left(1+t^\alpha\right)^2,
\end{equation*} \begin{equation*}
    2-2\left(\alpha-1\right)t^\alpha+\left(1+\mathfrak{E}_\alpha'(t)\right)t^{2\alpha}+2\left(1+\mathfrak{E}_\alpha'(t)\right)t^\alpha+\left(1+\mathfrak{E}_\alpha'(t)\right)=0,
\end{equation*} so we get the equation \begin{equation*}
   \left(1+\mathfrak{E}_\alpha'(t)\right) t^{2\alpha}+\left(4+2\mathfrak{E}_\alpha'(t)-2\alpha\right)t^\alpha+3+\mathfrak{E}_\alpha'(t)=0.
\end{equation*} $t=s_\alpha$ satisfies this equation, so solving for $s_\alpha^\alpha$, we get \begin{equation*}
    s_\alpha^\alpha=\frac{\alpha-2-\mathfrak{E}_\alpha'\left(s_\alpha\right)\pm \sqrt{\left(\alpha-2-\mathfrak{E}_\alpha'\left(s_\alpha\right)\right)^2-\left(1+\mathfrak{E}_\alpha'\left(s_\alpha\right)\right)\left(3+\mathfrak{E}_\alpha'\left(s_\alpha\right)\right)}}{1+\mathfrak{E}_\alpha'\left(s_\alpha\right)}. 
\end{equation*} Let \begin{equation}
    H_\alpha^{\pm}\left(x\right)=\frac{\alpha-2-x\pm \sqrt{\left(\alpha-2-x\right)^{2}-\left(1+x\right)\left(3+x\right)}}{1+x},
    \label{Hsubalphadef}
\end{equation} respectively so $s_\alpha^\alpha$ is one of $H_\alpha^+\left(\mathfrak{E}_\alpha'\left(s_\alpha\right)\right)$, $H_\alpha^{-}\left(\mathfrak{E}_\alpha'\left(s_\alpha\right)\right)$.   We can check that if $\alpha\ge 12$, \eqref{boundonmathfrakEderivative} implies $\left|\mathfrak{E}_\alpha'\left(s_\alpha\right)\right|\le 0.01$. In that case, we note if $|x| \le 0.01$, then \begin{align*}
    H_\alpha^{-}(x)&= \frac{3+x}{\alpha-2-x+\sqrt{\left(\alpha-2-x\right)^2-(1+x)(3+x)}} \\
    &\le \frac{3.01}{\alpha-2.01+\sqrt{\left(\alpha-2.01\right)^2-1.01\cdot 3.01}}<1.
\end{align*} Proposition \ref{thm:minimiserprop} says $s_\alpha> 1$, so this tells us for $\alpha\ge 12$, $s_\alpha^\alpha=H^{+}_{\alpha}\left(\mathfrak{E}_\alpha'\left(s_\alpha)\right)\right)$. We now bound the derivative of $H_\alpha^+$ for $|x|\le 0.01$. Note  \begin{align*}
    \left(\alpha-2-x\right)^{2}-\left(1+x\right)\left(3+x\right)&=\left(\alpha-2\right)^{2}-2\left(\alpha-2\right)x+x^{2}-x^{2}-4x-3 \\
    &=\left(\alpha-2\right)^{2}-2\alpha x-3,
\end{align*} so \begin{equation*}
    H_\alpha^+(x)=\frac{\sqrt{\left(\alpha-2\right)^{2}-2\alpha x-3}}{1+x}+\frac{\alpha-1}{1+x}-1.
\end{equation*}
Then \begin{align*}
    \partial_x H_\alpha^+(x)&=\frac{\frac{-\alpha\left(1+x\right)}{\sqrt{\left(\alpha-2\right)^2-2\alpha x-3}}-\sqrt{\left(\alpha-2\right)^2-2\alpha x-3}}{\left(1+x\right)^2}-\frac{\alpha-1}{\left(1+x\right)^2}, 
\end{align*} and if $|x|\le 0.01$, then \begin{align*}
    \left|\partial_x H_\alpha^+(x)\right|\le \frac1{0.99^2}\left(\alpha-1+\alpha-2+\frac{1.01\alpha}{\sqrt{(\alpha-2)^2-0.02\alpha-3}}\right)\le \frac{2\alpha}{0.99^2},
\end{align*} if $\alpha\ge 12$. So then \begin{equation*}
    \left|H_\alpha^+(x)-H_\alpha^+(0)\right|\le \frac{2\alpha}{0.99^2}|x|,
\end{equation*} and since $s_\alpha^\alpha=H_\alpha^+\mathopen{}\left(\mathfrak{E}_\alpha'\left(s_\alpha\right)\right)\mathclose{}$, we get \begin{equation*}
    \left|s_\alpha^\alpha-H_\alpha^+(0)\right|\le  \frac{2\alpha}{0.99^2}\left|\mathfrak{E}_\alpha'\left(s_\alpha\right)\right|.
\end{equation*} From \eqref{boundonmathfrakEderivative}, \begin{equation*}
    \left|\mathfrak{E}_\alpha'(t)\right|\le \frac{\left(\alpha+1\right)^2}{2^{\alpha-1}\left(\alpha-1\right)} ,
\end{equation*} if $t\ge 1$, and so at $t=s_\alpha$ \begin{equation}
    \left|s_\alpha^\alpha-H_\alpha^+(0)\right|\le \frac{2\alpha}{0.99^2}\cdot \frac{\left(\alpha+1\right)^2}{2^{\alpha-1}\left(\alpha-1\right)}.
    \label{paradis}
\end{equation} From \eqref{Hsubalphadef}, $ H_\alpha^+(0)=\alpha-2+\sqrt{\left(\alpha-2\right)^2-3}$, so this proves the first part of the proposition, \eqref{ssubalpjha^alphaeq}.  We use this to get estimates of $F_\alpha(1)$ and $F_\alpha'(1)$. For $\alpha\ge 12$, the bound in \eqref{paradis} implies \begin{align*}
    2\alpha-5 \le s_\alpha^\alpha\le 2\alpha-3.
\end{align*} Remember
$F_\alpha(x)=\left(1+s_\alpha^\alpha x^\alpha\right)^{-1}$, so $F_\alpha(1)=\left(1+s_\alpha^\alpha\right)^{-1}$. This immediately tells us that \begin{align*}
\frac1{2\alpha-2}\le F_\alpha(1)\le \frac1{2\alpha-4}.
\end{align*} From \eqref{emu}, $xF_\alpha'(x)=-\alpha F_\alpha(x)\left(1-F_\alpha(x)\right)$.
Plugging in $x=1$, $ -F_\alpha'(1)=\alpha F_\alpha(1)\left(1-F_\alpha(1)\right)$, so \begin{align*}
   \frac{1}{2}\left(1-\frac1{2\alpha-4}\right)\le -F_\alpha'(1)\le \frac{\alpha}{2\alpha-4}.
\end{align*} 
\end{proof}

\section{Constructing the auxiliary functions}\label{auxsec} 
With these facts, we can begin constructing the auxiliary functions satisfying Hypothesis \ref{thm:hypothesisonpsialpha}.
Recall we reduced the task of proving Theorem \ref{measure} to constructing an auxiliary even function $\psi_\alpha:\mathbb{R} \to \mathbb{R}$ satisfying Hypothesis \ref{thm:hypothesisonpsialpha}, restated below
\hypothesisonpsialpha*
 We use Theorem $9$ in \cite{bams/1183552525}, which allows us to construct a function $\psi_\alpha$ satisfying \eqref{fourthcondpsihatog} from its values and derivatives at $\mathbb{Z}$ (under certain conditions on $\psi_\alpha$), by  \begin{equation*}
    \psi_\alpha(x)=\frac{\sin^2\left(\pi x\right)}{\pi^2}\left(\sum_{n=-\infty}^\infty\frac{\psi_\alpha(n)}{\left(x-n\right)^2}+\sum_{n=-\infty}^\infty \frac{\psi_\alpha'(n)}{x-n}\right),
\end{equation*} and its Fourier transform (in $L^2$ sense) is supported in $[-1,1]$, given by for $|\xi|\le 1$, \begin{equation*}
\widehat{\psi_\alpha}(\xi)=\left(1-|\xi|\right)\left(\sum_{n=-\infty}^\infty \psi_\alpha(n)e^{-2\pi i n\xi}\right)+\frac{\sgn(\xi)}{2\pi i }\sum_{n=-\infty}^\infty \psi_\alpha'(n)e^{-2\pi i n\xi}.
\end{equation*} Equations \eqref{firstcondpsiog} and \eqref{seconcondpsiog} imply $\psi_\alpha(n)=F_\alpha(n)$ and $\psi_\alpha'(n)=F_\alpha'(n)$ for all $n \in \mathbb{Z}$, so we set \begin{equation}
\psi_\alpha(x)=\frac{\sin^2\left(\pi x\right)}{\pi^2}\left(\sum_{n=-\infty}^\infty\frac{F_\alpha(n)}{\left(x-n\right)^2}+\sum_{n=-\infty}^\infty \frac{F_\alpha'(n)}{x-n}\right).
    \label{psigensecconstr}
\end{equation} Taking advantage of the fact that $F_\alpha$ is even, we can write this as
\begin{equation}
\begin{split}
    \psi_\alpha(x)&= \sum_{n=-\infty}^\infty F_\alpha(n)\cdot \frac{\sin^2\left(\pi(x-n)\right)}{\pi^2\left(x-n\right)^2} \\
    &+\sum_{n=-\infty}^\infty nF_\alpha'(n)\cdot \frac{\sin\left(\pi(x-n)\right)}{\pi(x-n)}\cdot \frac{\sin\left(\pi(x+n)\right)}{\pi(x+n)} ,
\end{split}
\label{hauptul}
\end{equation}
from which we can see that the series converges in $L^p$, $p\ge 1$. Consequently, we can take the Fourier transform termwise, which leads to the following formula for $\widehat{\psi_\alpha}$: $\widehat{\psi_\alpha}(\xi)=0$ if $|\xi|>1$, and if $|\xi|\le 1$,  
\begin{equation}
\widehat{\psi_\alpha}(\xi)=\left(1-|\xi|\right)\left(\sum_{n=-\infty}^\infty F_\alpha(n)e^{-2\pi i n\xi}\right)+\frac{\sgn(\xi)}{2\pi i }\sum_{n=-\infty}^\infty F_\alpha'(n)e^{-2\pi i n\xi}.
    \label{widehatpsiin-1,1}
\end{equation}
In particular, $\psi_\alpha$ satisfies \eqref{fourthcondpsihatog}, and the expression in \eqref{psigensecconstr} implies $\psi_\alpha$ satisfies \eqref{seconcondpsiog}. We now use  the expression in \eqref{psigensecconstr} to show \eqref{sev}:
\begin{proof}[Proof of \eqref{sev}]
    Multiplying both sides of \eqref{hauptul} by $x^2$, we have \begin{align*}
        x^2\psi_\alpha(x)&=\sum_{n=-\infty}^\infty F_\alpha(n)\left(\frac{x\sin\left(\pi (x-n)\right)}{\pi\left(x-n\right)}\right)^2+\sum_{n=-\infty}^\infty nF_\alpha'(n)\frac{x^2\sin^2\left(\pi x\right)}{\pi^2\left(x^2-n^2\right)}.
    \end{align*} Note \begin{align*}
        \left|\frac{x\sin\left(\pi (x-n)\right)}{\pi\left(x-n\right)}\right|=\left|\frac{\sin\left(\pi(x-n)\right)}{\pi}+n\frac{\sin\left(\pi(x-n)\right)}{\pi(x-n)}\right|\le \frac1{\pi}+|n|,
    \end{align*} and \begin{align*}
        \left|\frac{x^2\sin^2\left(\pi x\right)}{\pi^2\left(x^2-n^2\right)}\right|&=\left|\frac{\sin^2\left(\pi x\right)}{\pi^2}+n^2\frac{\sin\left(\pi(x-n)\right)}{\pi(x-n)}\cdot \frac{\sin\left(\pi(x+n)\right)}{\pi(x+n)}\right|\le \frac1{\pi^2}+n^2,
    \end{align*} so \begin{align*}
        \left|x^2\psi_\alpha(x)\right|\le \sum_{n \in \mathbb{Z}}F_\alpha(n)\left(\frac1{\pi}+|n|\right)^2+\sum_{n \in \mathbb{Z}}n\left|F_\alpha'(n)\right|\left(\frac1{\pi^2}+n^2\right).
    \end{align*} Since $\alpha\ge 4$, and $F_\alpha(x)=\frac1{s_\alpha^\alpha x^\alpha+1}$ for some $s_\alpha>0$, this implies the series above is finite. This completes the proof.
\end{proof}
So far, we have shown $\psi_\alpha$ satisfies \eqref{seconcondpsiog},\eqref{fourthcondpsihatog} and \eqref{sev}. What remains is to show the inequalities in \eqref{firstcondpsiog}, \eqref{thirdcondpsihatog}, so showing $\widehat{\psi_\alpha}(\xi)\ge 0$, and $\psi_\alpha(x)\le F_\alpha(x)$. We do these in sections \ref{widehatpsi-alphapos} and \ref{chamere} respectively.

\section{The non-negativity of \texorpdfstring{$\widehat{\psi_\alpha}$}{Lg}.}\label{widehatpsi-alphapos}
Recall $\widehat{\psi_\alpha}$ given by \eqref{widehatpsiin-1,1}, \begin{equation*}
    \widehat{\psi_\alpha}(\xi)=\left(1-|\xi|\right)\left(\sum_{n=-\infty}^\infty F_\alpha(n)e^{-2\pi i n\xi}\right)+\frac{\sgn(\xi)}{2\pi i }\sum_{n=-\infty}^\infty F_\alpha'(n)e^{-2\pi i n\xi},
\end{equation*} if $|\xi|\le 1$, and $0$ otherwise. We already know that $\widehat{\psi_\alpha}$ is even since $\psi_\alpha$ is even and real-valued, so we will restrict our attention to $\xi \in [0,1]$. We start with when $0\le \xi\le \frac1{2}$:
\subsection{Proof that \texorpdfstring{$\widehat{\psi_\alpha}(\xi)\ge 0$ when $\xi \in \left[0,\frac1{2}\right]$}{Lg}}
We symmetrise the above expression for $\widehat{\psi_\alpha}$ to get \begin{equation}
    \widehat{\psi_\alpha}(\xi)=\left(1-\xi\right)\sum_{n=-\infty}^\infty F_\alpha(n)\cos\left(2\pi n\xi\right)-\frac1{2\pi}\sum_{n=-\infty}^\infty F_\alpha'(n)\sin\left(2\pi n \xi\right),
    \label{billieeilish}
\end{equation} so 
\begin{equation*}
    \widehat{\psi_\alpha}(\xi)\ge \left(1-\xi\right)\left(1-\sum_{n=1}^\infty 2F_\alpha(n)\right)-\frac1{\pi}F_\alpha'(1)\sin\left(2\pi \xi\right)-\frac1{\pi}\sum_{n=2}^\infty \left|F_\alpha'(n)\right|.
\end{equation*} Note that $F_\alpha'(1)<0$ and if $0\le \xi\le \frac1{2}$, $\sin\left(2\pi  \xi\right)\ge 0$, so then if $0\le \xi\le \frac1{2}$, then \begin{equation}
    \widehat{\psi_\alpha}(\xi)\ge \frac1{2}\left(1-\sum_{n=1}^\infty 2F_\alpha(n)\right)-\frac1{\pi}\sum_{n=2}^\infty \left|F_\alpha'(n)\right|.  \label{crookes2}
\end{equation} Call the constant on the right-hand side $\mathcal{T}(\alpha)$, so it is sufficient to show $\mathcal{T}(\alpha)\ge 0$ to prove $\widehat{\psi_\alpha}(\xi)\ge 0$ for $0\le \xi \le \frac1{2}$. When $4\le \alpha\le 10$, we verify this via interval arithmetic, the code for which can be found in the following Github repository \cite{edwin2024}. When $\alpha\ge 12$, we prove this directly:
\begin{proof}[Proof of $\mathcal{T}(\alpha)\ge 0$ for $\alpha\ge 12$.]
     From Proposition \ref{thm:asympsalpha}, we have $F_\alpha(1)\le \frac1{2\alpha-4}$, and $F_\alpha(n)\le \alpha^{-1}n^{-\alpha}$. From \eqref{emu}, we also have \begin{align*}
         \left|F_\alpha'(n)\right|\le \frac{\alpha F_\alpha(n)}{n}\le \frac1{n^{\alpha+1}},
     \end{align*}
     since as we said earlier, $F_\alpha(n)\le \alpha^{-1}n^{-\alpha}$, and so  \begin{align*}
    \mathcal{T}(\alpha)\ge \frac1{2}\left(1-\frac1{\alpha-2}-\sum_{n=2}^\infty\frac2{\alpha n^\alpha}\right)-\frac1{\pi}\sum_{n=2}^\infty \frac1{n^{1+\alpha}}.
\end{align*} Using Lemma \ref{estimateforsumoverpowers}, this implies \begin{align*}
    \mathcal{T}(\alpha)\ge \frac1{2}\left(1-\frac1{\alpha-2}-\frac2{\alpha}\cdot \frac{\alpha+1}{ 2^\alpha(\alpha-1)}\right)-\frac1{\pi}\cdot \frac{\alpha+2}{\alpha 2^{1+\alpha}},
\end{align*} which is positive if $\alpha\ge 12$. 
\end{proof}
We now move on to the case where $\xi \in \left[\frac1{2},1\right]$.
\subsection{Proof that \texorpdfstring{$\widehat{\psi_\alpha}(\xi)\ge 0$ when $\xi \in \left[\frac1{2},1\right]$}{Lg}}
 Write $\xi=1-t$ with $t \in \left[0,\frac1{2}\right]$ and substitute this into the expression for $\widehat{\psi_\alpha}$ in \eqref{billieeilish}, so that \begin{align*}
    \widehat{\psi_\alpha}(1-t)&=t\sum_{n=-\infty}^\infty F_\alpha(n)\cos\left(2\pi nt\right)+\frac1{2\pi}\sum_{n=-\infty}^\infty F_\alpha'(n)\sin\left(2\pi n t\right) \\
    &=t\left(\sum_{n=-\infty}^\infty F_\alpha(n)\cos\left(2\pi n t\right)+\sum_{n=-\infty}^\infty n F_\alpha'(n)\frac{\sin(2\pi n t)}{2\pi nt}\right).
\end{align*}  From \eqref{accfirstordercond}, $  \sum_{n \in \mathbb{Z}} \left(F_\alpha(n)+nF_\alpha'(n)\right)=0$, so subtracting this sum from the right-hand side of the above equation, we get \begin{align*}
    \widehat{\psi_\alpha}(1-t)&=t\sum_{n=-\infty}^\infty\mathopen{}\left( F_\alpha(n)\mathopen{}\left(\cos\left(2\pi n t\right)-1\right)\mathclose{}+ n F_\alpha'(n)\mathopen{}\left(\frac{\sin(2\pi n t)}{2\pi nt}-1\right)\mathclose{}\right)\mathclose{}, 
\end{align*} and using the half angle formula for cosine, we get \begin{align}
    \widehat{\psi_\alpha}(1-t)&=t\sum_{n=-\infty}^\infty\left( n F_\alpha'(n)\left(\frac{\sin(2\pi n t)}{2\pi nt}-1\right)-t^2\pi^2\cdot 2n^2F_\alpha(n)\frac{\sin^2\left(\pi n t\right)}{\left(\pi n t\right)^2}\right). 
    \label{beforehours}
\end{align}
We make use of the following lemma:
\begin{lemma}
For $x \in \mathbb{R}$ define $\mathcal{R}(x)$ for $x\neq 0$ by \begin{align*}
    \sin x=x-\frac{x^3}{6}+x^3\mathcal{R}(x)
\end{align*} and set $\mathcal{R}(0)=0$. Then $\mathcal{R}$ is continuous and increasing in $|x|$.
    \label{mathfrakRdeflemma}
\end{lemma}
\begin{proof}[Proof of Lemma \ref{mathfrakRdeflemma}]
That $\mathcal{R}$ is continuous follows from the taylor series of $\sin x$. Note $\mathcal{R}$ is even, so it is sufficient to show  $\mathcal{R}'(x)\ge 0$ for $x\ge 0$.
    By differentiating, \begin{align*}
    \mathcal{R}'(x)=\frac{x^3\left(\cos x-1\right)-3x^2\left(\sin x-x\right)}{x^6},
\end{align*} so it is equivalent to show \begin{align*}
    h(x):=x\cos x+2x-3\sin x\ge 0.
\end{align*} By using the half-angle formulae for sine and cosine, we get \begin{align*}
    h'(x)&=4\sin^2\left(\frac{x}{2}\right)-x\sin x=4\sin\left(\frac{x}{2}\right)\cos\left(\frac{x}{2}\right)\left(\tan\left(\frac{x}{2}\right)-\frac{x}{2}\right) \\
    &=2\sin(x)\left(\tan\left(\frac{x}{2}\right)-\frac{x}{2}\right)\ge 0
\end{align*}  if $x \in [0,\pi]$, and so $h$ is increasing on $[0,\pi]$. $h(0)=0$, so this shows $h(x)\ge 0$ if $x \in [0,\pi]$. If $x\ge \pi $, then $h(x)\ge \pi-3\ge 0$, and so $h\ge 0$. This completes the proof.
\end{proof}
From the definition of $\mathcal{R}$ in Lemma \ref{mathfrakRdeflemma}, we write\begin{equation*}
    \frac{\sin\left(2\pi n t\right)}{2\pi nt}-1=-\frac{\left(2\pi  n t\right)^2}{6}+\left(2\pi nt\right)^2\mathcal{R}\left(2\pi nt\right),
\end{equation*} so \begin{equation*}
    \begin{split}
        &\sum_{n=-\infty}^\infty n F_\alpha'(n)\left(\frac{\sin(2\pi n t)}{2\pi nt}-1\right) \\
    &=\pi^2t^2\sum_{n=-\infty}^\infty nF_\alpha'(n)\left(\left(2n\right)^2\mathcal{R}\left(2\pi n t\right)-\frac{2n^2}{3}\right).
    \end{split}
\end{equation*} Plugging this into \eqref{beforehours} gives \begin{equation}
    \frac{\widehat{\psi_\alpha}(1-t)}{\pi^2t^3}=\sum_{n=-\infty}^\infty n^3F_\alpha'(n)\left(-\frac{2}{3}+4\mathcal{R}\left(2\pi n t\right)\right)-\sum_{n=-\infty}^\infty2n^2F_\alpha(n)\frac{\sin^2\left(\pi n t\right)}{\left(\pi n t\right)^2}.
    \label{niceidentity}
\end{equation} So our task is in showing the right-hand is non-negative, for $0\le t\le \frac1{2}$. For $\alpha=4$, we know $s_4=\sqrt{2}$ from Proposition \ref{thm:s4issqrt2prop}, and so we can readily verify this inequality.
\begin{proof}[Proof of \eqref{niceidentity} for $\alpha=4$]
     From \eqref{emu}, $xF_4'(x)=-4F_4(x)\left(1-F_4(x)\right)$, and so $-x^3F_4'(x)=4x^2F_4(x)\left(1-F_4(x)\right)$, so the right-hand side of \eqref{niceidentity} when $\alpha=4$ can be written as \begin{align*}
        \sum_{n=-\infty}^\infty n^2F_4(n)\left(4\left(1-F_4(n)\right)\left(\frac2{3}-4\mathcal{R}\left(2\pi n t\right)\right)-\frac{2\sin^2\left(\pi n t\right)}{\left(\pi n t\right)^2}\right).
    \end{align*} Note $1-F_4(n)\ge \frac4{5}$ for $n\ge 1$, so to prove the lemma, it is sufficient to show \begin{align*}
        \frac{16}5\left(\frac1{3}-2\mathcal{R}(2\pi nt)\right)\underset{\textup{wts}}{\ge} \frac{\sin^2\left(\pi n t\right)}{\left(\pi n t\right)^2}
    \end{align*}for $t \in \left[0,\frac1{2}\right]$, $n\in \mathbb{N}$. Making the substitution $w=\pi nt$ and substituting the expression for $\mathcal{R}$, this inequality simplifies to \begin{align*}
        \frac{4}{5}\cdot\frac{2w-\sin\left(2w\right)}{w^{3}}\underset{\textup{wts}}{\ge} \frac{\sin^2(w)}{w^2},
    \end{align*} for $w\ge 0$, or\begin{align*}
        8w-4\sin(2w)-5w\sin^2(w)\ge 0.
    \end{align*} By differentiating and simplifying, we get \begin{align*}
        \partial_w\left(8w-4\sin(2w)-5w\sin^2(w)\right)&=\sin w\left(11\sin w-10w\cos w\right) \\
        &\ge 10(\sin w)(\cos w)\left(\tan w-w\right)\ge 0
    \end{align*} if $w \in \left[0,\frac{\pi}{2}\right]$, which implies $8w-4\sin(2w)-5w\sin^2(w)\ge 0$ if $w \in \left[0,\frac{\pi}{2}\right]$. If $w\ge \frac{\pi}{2}$, then \begin{align*}
        8w-4\sin(2w)-5w\sin^2(w)\ge 3\cdot \frac{\pi}{2}-4\ge 0,
    \end{align*}  so this shows $ 8w-4\sin(2w)-5w\sin^2(w)\ge 0$, as desired.
\end{proof}
We now focus on $\alpha\ge 6$.
For $\alpha\ge 6$, $t\in \left[0,\frac{1}{2}\right]$, since $\mathcal{R}(x)$ is increasing in $|x|$, \begin{align*}
    -\frac2{3}+4\mathcal{R}\left(2\pi n t\right)\le -\frac2{3}+4\mathcal{R}\left(\pi n \right),
\end{align*} and $n^3F_\alpha'(n)\le 0$, so from \eqref{niceidentity}, we get \begin{equation}
    \frac{\widehat{\psi_\alpha}(1-t)}{\pi^2t^3}\ge \sum_{n=-\infty}^\infty n^3F_\alpha'(n)\left(-\frac{2}{3}+4\mathcal{R}\left(\pi n \right)\right)-\sum_{n=-\infty}^\infty 2n^2F_\alpha(n). \label{loreen}
\end{equation}
Call the constant on the right-hand side $\mathfrak{L}(\alpha)$, so to show $\widehat{\psi_\alpha}(\xi)\ge 0$ for $\xi \in \left[\frac1{2},1\right]$, it is sufficient to show $\mathfrak{L}(\alpha)\ge 0$. When $6\le \alpha\le 10$, we verify this via interval arithmetic, the code for which can be found in the following Github repository \cite{edwin2024}. We now prove it for $\alpha\ge 12$:
\begin{proof}[Proof of $\mathfrak{L}(\alpha)\ge 0$ for $\alpha\ge 12$]
    We write \begin{align*}
    \mathfrak{L}(\alpha)&=-2F_\alpha'(1)\left(\frac2{3}-4\mathcal{R}(\pi)\right)+2\sum_{n=2}^\infty n^3F_\alpha'(n)\left(-\frac2{3}+4\mathcal{R}(\pi n)\right) \\
    &-4F_\alpha(1)-2\sum_{n=2}^\infty 2n^2F_\alpha(n).
\end{align*} From \eqref{boundonFalphaof1} and \eqref{boundFalpha'} (recall $\mathcal{R}(x)$ is increasing in $|x|$) this implies\begin{align*}
    \mathfrak{L}(\alpha)&\ge \left(1-\frac1{2\alpha-4}\right)\left(\frac2{3}-4\mathcal{R}(\pi)\right)-2\sum_{n=2}^\infty n^3F_\alpha'(n)\left(\frac2{3}-4\mathcal{R}(\pi n)\right) \\
    &-\frac{2}{\alpha-2}-2\sum_{n=2}^\infty 2n^2F_\alpha(n).
\end{align*} Note for $n\ge 0$, $-n^3F_\alpha'(n)\ge0$, and $0\le \mathcal{R}(x)\le \frac1{6}$ from Lemma \ref{mathfrakRdeflemma}, so we get \begin{align*}
    \mathfrak{L}(\alpha)\ge \left(1-\frac1{2\alpha-4}\right)\left(\frac2{3}-4\mathcal{R}(\pi)\right)-4\sum_{n=2}^\infty n^2F_\alpha(n)-\frac2{\alpha-2}.
\end{align*} We also know $F_\alpha(n)\le (2\alpha-5)^{-1}n^{-\alpha}$ for $\alpha\ge 12$ from Proposition \ref{thm:asympsalpha}, and combining this with Lemma \ref{estimateforsumoverpowers}, we get \begin{align*}
    \sum_{n=2}^\infty n^2F_\alpha(n)\le \frac1{2\alpha-5}\sum_{n=2}^\infty \frac1{n^{\alpha-2}}\le \frac1{2\alpha-5}\cdot \frac1{2^{\alpha-2}}\left(\frac{\alpha-1}{\alpha-3}\right),
\end{align*} which implies if $\alpha\ge 12$, then \begin{align*}
    \mathfrak{L}(\alpha)\ge \left(1-\frac1{2\alpha-4}\right)\left(\frac2{3}-4\mathcal{R}(\pi)\right)-\frac4{2^{\alpha-2}(2\alpha-5)}\cdot\frac{\alpha-1}{\alpha-3}-\frac2{\alpha-2}\ge 0,
\end{align*} as desired.
\end{proof} 
We now go on to proving $\psi_\alpha(x)\le F_\alpha(x)$:

\section{Proving \texorpdfstring{$\psi_\alpha(x)\le F_\alpha(x)$}{Lg}}\label{chamere}
We want to show $\psi_\alpha(x)\le F_\alpha(x)$, and using the expression for $\psi_\alpha(x)$ in \eqref{psigensecconstr}, it is equivalent to show \begin{equation}
    \frac{\sin^2\left(\pi x\right)}{\pi^2}\left(\sum_{n=-\infty}^\infty \frac{F_\alpha(n)}{\left(x-n\right)^2}+\sum_{n=-\infty}^\infty \frac{F_\alpha'(n)}{x-n}\right)\underset{\textup{wts}}{\le} F_\alpha(x) \label{andre}
\end{equation} for $x\ge 0$, since $\psi_\alpha$ and $F_\alpha$ are even. We first focus on the $\alpha=4$ case, and then move on to $\alpha\ge 6$.\newline

\noindent
\newline

\noindent
\subsection{Proving\texorpdfstring{ $\psi_4(x)\le F_4(x)$}{Lg}}
From Proposition \ref{thm:minimiserprop}, $s_4=\sqrt{2}$, so $F_4(x)=\left(4x^4+1\right)^{-1}$, so the inequality we want to show in \eqref{andre} is \begin{align}
     \frac{\sin^2\left(\pi x\right)}{\pi^2}\left(\sum_{n=-\infty}^\infty \frac{F_4(n)}{\left(x-n\right)^2}+\sum_{n=-\infty}^\infty \frac{F_4'(n)}{x-n}\right)\underset{\textup{wts}}{\le} F_4(x), \label{gagas}
\end{align} for $x\ge 0$. The way we do this is show this holds for $x\ge 9$, and verify it numerically on $[0,9]$. Throughout this part, we make use of the identity \begin{align}
\frac{\pi^2}{\sin^2\left(\pi x\right)}=\sum_{n \in \mathbb{Z}}\frac1{\left(x-n\right)^2}.
    \label{sumover1oversin^2}
\end{align} 
\begin{proof}[Proof of $\psi_4(x)\le F_4(x)$ for $x\ge 9$]
    using the identity \eqref{sumover1oversin^2}, we multiply both sides of \eqref{gagas} by $\frac{\pi^2}{\sin^2(\pi x)}$ to get the equivalent inequality 
    \begin{align}
        \sum_{n \in \mathbb{Z}}\left(\frac{F_4(x)-F_4(n)-F_4'(n)(x-n)}{(x-n)^2}\right)\underset{\textup{wts}}{\ge} 0. 
        \label{hae}
    \end{align} Let $\eta(x)$ be the integer closest to $x$ (rounded up if $x$ is halfway between two integers). By the mean value theorem, \begin{align*}
       \frac{F_4(x)-F_4(\eta(x))-F_4'(\eta(x))(x-\eta(x))}{(x-\eta(x))^2}=\frac1{2}F_4''(s)
    \end{align*} for some $s$ between $x$ and $\eta(x)$. If $x\ge 9$, any such $s$ satisfies $F_4(s)\le \frac1{2}$, and from \eqref{Falphaduobleprime}, this implies $F_4''(s)\ge 0$. So it is sufficient to show
    \begin{align}
    \sum_{n \in \mathbb{Z}\setminus\{\eta(x)\}}\left(\frac{F_4(x)-F_4(n)-F_4'(n)(x-n)}{(x-n)^2}\right)\underset{\textup{wts}}{\ge} 0
    \label{hae2}
\end{align} for $x\ge 9$. We pull out the $n=-\eta(x)$ term, so it is equivalent to show \begin{equation}
    \begin{split}
        &\sum_{\substack{n \in \mathbb{Z} \\ |n|\neq \eta(x)}}\left(\frac{F_4(x)-F_4(n)-F_4'(n)(x-n)}{(x-n)^2}\right)\underset{\textup{wts}}{\ge} \\
        &-\frac{F_4(x)-F_4(\eta(x))+F_4'(\eta(x))(x+\eta(x))}{(x+\eta(x))^2}.
    \end{split}
    \label{watermelonsugar}
\end{equation} We start by analysing the left-hand side, for general $\alpha$ which will be useful later. Let \begin{align}
    L_\alpha(x,n)=\frac{F_\alpha(x)-F_\alpha(n)-F_\alpha'(n)(x-n)}{(x-n)^2},
    \label{Lxndef}
\end{align} and note \begin{align*}
    &\frac{x^2\left(L_\alpha(x,n)+L_\alpha(x,-n)\right)}{2} \\
    &=x^2\left(\left(F_\alpha(x)-F_\alpha(n)\right)\left(\frac{x^2+n^2}{\left(x^2-n^2\right)^2}\right)-\frac{nF_\alpha'(n)}{x^2-n^2}\right) \\
    &= \left(F_\alpha(x)-F_\alpha(n)\right)\left(\frac{x^4+n^2x^2}{\left(x^2-n^2\right)^2}\right)-\frac{nx^2F_\alpha'(n)}{x^2-n^2} \\
    &=-F_\alpha(n)-nF_\alpha'(n)+F_\alpha(x)\left(\frac{x^4+n^2x^2}{\left(x^2-n^2\right)^2}\right) \\
    &-F_\alpha(n)\left(\frac{3n^2x^2-n^4}{\left(x^2-n^2\right)^2}\right)-\frac{n^3F_\alpha'(n)}{x^2-n^2},
\end{align*} and multiplying both sides by $x^2$ again, \begin{equation*}
         \begin{split}
             &\frac{x^4\left(L_\alpha(x,n)+L_\alpha(x,-n)\right)}{2} \\
     &=-x^2\left(F_\alpha(n)+nF_\alpha'(n)\right)+x^4F_\alpha(x)\left(\frac{x^2+n^2}{\left(x^2-n^2\right)^2}\right) \\
     &-F_\alpha(n)\left(\frac{3n^2x^4-n^4x^2}{\left(x^2-n^2\right)^2}-3n^2+3n^2\right)-\frac{n^3x^2F_\alpha'(n)}{x^2-n^2} \\
     &=-x^2\left(F_\alpha(n)+nF_\alpha'(n)\right)+\frac{x^4F_\alpha(x)\left(x^2+n^2\right)-n^4F_\alpha(n)\left(5x^2-3n^2\right)}{\left(x^2-n^2\right)^2} \\
     &-3n^2F_\alpha(n)-n^3F_\alpha'(n)-\frac{n^5F_\alpha'(n)}{x^2-n^2} \\
     &=-x^2\left(F_\alpha(n)+nF_\alpha'(n)\right)+\frac{x^4F_\alpha(x)\left(x^2+n^2\right)-2n^6F_\alpha(n)}{\left(x^2-n^2\right)^2} \\
     &-3n^2F_\alpha(n)-n^3F_\alpha'(n)-\frac{5n^4F_\alpha(n)+n^5F_\alpha'(n)}{x^2-n^2}.
         \end{split}
\end{equation*} Summing both sides over $n \in \mathbb{Z}\setminus\{\eta(x),-\eta(x)\}$, using the definition of $L_\alpha(x,n)$ in \eqref{Lxndef} and applying the identity in \eqref{accfirstordercond}, this implies 
\begin{equation}
\begin{split}
     &x^4  \sum_{\substack{n \in \mathbb{Z} \\ |n|\neq \eta(x)}}\left(\frac{F_\alpha(x)-F_\alpha(n)-F_\alpha'(n)(x-n)}{(x-n)^2}\right) \\
    &= 2x^2\left(F_\alpha(\eta(x))+\eta(x)F_\alpha'(\eta(x))\right)-\sum_{\substack{n \in \mathbb{Z} \\ |n|\neq \eta(x)}}\left(3n^2F_\alpha(n)+n^3F_\alpha'(n)\right) \\
    &+\sum_{\substack{n \in \mathbb{Z} \\ |n|\neq \eta(x)}}\left(\frac{x^4F_\alpha(x)\left(x^2+n^2\right)-2n^6F_\alpha(n)}{\left(x^2-n^2\right)^2}-\frac{5n^4F_\alpha(n)+n^5F_\alpha'(n)}{x^2-n^2}\right).
\end{split}
    \label{ipadair}
\end{equation} This identity will be useful later when working in the $\alpha\ge 6$ case. Now, some algebra. Using the identity $x^4F_4(x)=\frac1{4}-\frac{F_4(x)}{4}$,$n^4F_4(n)=\frac1{4}-\frac{F_4(n)}{4}$ we see that 
\begin{equation}
\begin{split}
     &\frac{x^4F_4(x)\left(x^2+n^2\right)-2n^6F_4(n)}{\left(x^2-n^2\right)^2}-\frac{5n^4F_4(n)+n^5F_4'(n)}{x^2-n^2} \\
&=\frac{\frac1{4}-5n^4F_4(n)-n^5F_4'(n)}{x^2-n^2}+\frac{\frac{n^2F_4(n)}{2}-\frac{F_4(x)\left(x^2+n^2\right)}{4}}{\left(x^2-n^2\right)^2}.
\end{split}
    \label{Stella}
\end{equation} 
Evaluating the numerator of the first term of the right-hand side, we get \begin{align*}
    \frac1{4}-5n^4F_4(n)-n^5F_4'(n)&=-\frac{3n^4-\frac1{4}}{\left(4n^4+1\right)^2}.
\end{align*} Moreover, if $|n|\neq \eta(x)$, then \begin{align*}
    \frac1{\left|x^2-n^2\right|}=\frac1{x^2}\left|1+\frac{n^2}{x^2-n^2}\right|\le \frac1{x^2}\left|1+\frac{n^2}{\left|x-n\right||n|}\right|\le \frac{1+2|n|}{x^2}.
\end{align*}  From \eqref{Stella}, this implies \begin{align*}
    & \frac{x^4F_4(x)\left(x^2+n^2\right)-2n^6F_4(n)}{\left(x^2-n^2\right)^2}-\frac{5n^4F_4(n)+n^5F_4'(n)}{x^2-n^2}\ge \\
    &-\frac{\left(3n^4-\frac1{4}\right)(2|n|+1)}{\left(4n^4+1\right)^2x^2}-\frac{F_4(x)}{4}\cdot \frac{x^2+n^2}{\left(x^2-n^2\right)^2},
\end{align*} and so \begin{align*}
   & \sum_{\substack{n \in \mathbb{Z} \\ |n|\neq \eta(x)}}\left(\frac{x^4F_4(x)\left(x^2+n^2\right)-2n^6F_4(n)}{\left(x^2-n^2\right)^2}-\frac{5n^4F_4(n)+n^5F_4'(n)}{x^2-n^2}\right) \\
    &\ge -\frac1{x^2}\sum_{n \in \mathbb{Z}}\frac{\left(3n^4-\frac1{4}\right)(2|n|+1)}{\left(4n^4+1\right)^2}-\frac{F_4(x)}{4}\sum_{\substack{n \in \mathbb{Z} \\ |n|\neq \eta(x)}} \frac{x^2+n^2}{\left(x^2-n^2\right)^2} \\
    &\ge -\frac1{x^2}-\frac{F_4(x)}{4}\sum_{\substack{n \in \mathbb{Z} \\ |n|\neq \eta(x)}} \frac{x^2+n^2}{\left(x^2-n^2\right)^2}.
\end{align*}
Bounding the second term on the right-hand side, \begin{equation}
    \begin{split}
        \sum_{\substack{n \in \mathbb{Z} \\ |n|\neq \eta(x)}}\frac{x^2+n^2}{\left(x^2-n^2\right)^2}&=\frac1{2}\sum_{\substack{n \in \mathbb{Z} \\ |n|\neq \eta(x)}}\left(\frac1{(x-n)^2}+\frac1{(x+n)^2}\right) \\
    &\le \sum_{k=0}^\infty\frac2{\left(k+\frac1{2}\right)^2}\le 10,
    \end{split}
    \label{obrien}
\end{equation} which implies \begin{align*}
   &\sum_{\substack{n \in \mathbb{Z} \\ |n|\neq \eta(x)}}\left(\frac{x^4F_4(x)\left(x^2+n^2\right)-2n^6F_4(n)}{\left(x^2-n^2\right)^2}-\frac{5n^4F_4(n)+n^5F_4'(n)}{x^2-n^2}\right) \\
   &\ge -\frac1{x^2}-\frac{5F_4(x)}{2}.
\end{align*} Substituting this into \eqref{ipadair} implies \begin{align*}
   & x^4\sum_{\substack{n \in \mathbb{Z} \\ |n|\neq \eta(x)}}\left(\frac{F_4(x)-F_4(n)-F_4'(n)(x-n)}{(x-n)^2}\right) \\
   &\ge 2x^2\eta(x)F_4'(\eta(x))-\sum_{\substack{n \in \mathbb{Z} \\ |n|\neq \eta(x)}}\left(3n^2F_4(n)+n^3F_4'(n)\right)-\frac1{x^2}-\frac{5F_4(x)}{2}.
\end{align*}
Going back to \eqref{watermelonsugar}, it is sufficient to show that \begin{equation}
\begin{split}
     &2x^2\eta(x)F_4'(\eta(x))-\sum_{\substack{n \in \mathbb{Z} }}\left(3n^2F_4(n)+n^3F_4'(n)\right)-\frac1{x^2}-\frac{5F_4(x)}{2} \\
     &\underset{\textup{wts}}{\ge} -\frac{x^4\left(F_4(x)-F_4(\eta(x))+F_4'(\eta(x))(x+\eta(x))\right)}{(x+\eta(x))^2} \\
    & -6\eta(x)^2F_4(\eta(x))-2\eta(x)^3F_4'(\eta(x)).
\end{split}
    \label{onerightnow}
\end{equation} At this point we are almost done. Differentiating $F_4$, we see that \begin{align*}
    -F_4'(x)=\frac{16x^3}{\left(4x^4+1\right)^2}\le \frac1{x^5},
\end{align*} and from this, we may deduce for $x\ge 9$, \begin{align*}
    & -2\left(x^2\eta(x)+\eta(x)^3\right)F_4'(\eta(x))-6\eta(x)^2F_4(\eta(x)) \\
    &-\frac{x^4\left(F_4(x)-F_4(\eta(x))+F_4'(\eta(x))(x+\eta(x))\right)}{(x+\eta(x))^2} \\
    &\le -6\eta(x)^3F_4'(\eta(x))+\frac{x^4}{(x+\eta(x))^2}\left(F_4(\eta(x))-F_4'(\eta(x))(x+\eta(x))\right) \\
    &\le -8\eta(x)^3F_4'(\eta(x))+2\eta(x)^2F_4(\eta(x))\le \frac{8}{\eta(x)^2}+\frac1{2\eta(x)^2}\le \frac{10}{x^2},
\end{align*} so from \eqref{onerightnow}, it suffices to show \begin{align*}
    -\sum_{\substack{n \in \mathbb{Z} }}\left(3n^2F_4(n)+n^3F_4'(n)\right)-\frac1{x^2}-\frac{5F_4(x)}{2}\ge \frac{10}{x^2}
\end{align*} for $x\ge 9$, which follows from the fact that \begin{align*}
    -\sum_{\substack{n \in \mathbb{Z} }}\left(3n^2F_4(n)+n^3F_4'(n)\right)\ge \frac{10}{81}+\frac1{81}+\frac{5F_4(9)}{2}.
\end{align*}
\end{proof} It thus remains to show $\psi_4(x)\le F_4(x)$ for $0\le x\le 9$, and we do this by showing \eqref{hae}. For referencing we put this in the following lemma:
\begin{restatable}{lemma}{psifourleFfourx}
    Let $F_4(x)=\left(4x^4+1\right)^{-1}$. Then for $0\le x\le 9$, \begin{align*}
        \sum_{n \in \mathbb{Z}}\left(\frac{F_4(x)-F_4(x)-F_4'(n)(x-n)}{(x-n)^2}\right)\ge 0.
    \end{align*}
    \label{thm:psifourleFourx}
\end{restatable} The verification for this inequality and others can be found in the following Github repository \cite{edwin2024}.
We now move onto $\alpha\ge 6$ case. Recall the inequality we want to show, \eqref{andre}:
\begin{align*}
    \frac{\sin^2\left(\pi x\right)}{\pi^2}\left(\sum_{n=-\infty}^\infty \frac{F_\alpha(n)}{\left(x-n\right)^2}+\sum_{n=-\infty}^\infty \frac{F_\alpha'(n)}{x-n}\right)\underset{\textup{wts}}{\le} F_\alpha(x).
\end{align*}
As before, let $\eta(x)$ be the closest integer to $x$. We write $x=\eta(x)+t$, $|t|\le \frac1{2}$.
We write the left-hand side of the above inequality as 
\begin{equation*}
     \frac{\sin^2\left(\pi t\right)}{\pi^2}\sum_{n \in \mathbb{Z}\setminus\{\eta(x)\}}\left(\frac{F_\alpha(n)}{\left(x-n\right)^2}+\frac{F_\alpha'(n)}{x-n}\right)+\frac{\sin^2\left(\pi t\right)}{\pi^2t^2}\left(F_\alpha\left(\eta(x)\right)+tF_\alpha'\left(\eta(x)\right)\right),
\end{equation*} so we want to show \begin{equation}
    \begin{split}
        &F_\alpha(x)-\frac{\sin^2\left(\pi t\right)}{\pi^2t^2}\left(F_\alpha\left(\eta(x)\right)+tF_\alpha'\left(\eta(x)\right)\right) \\
        &\underset{\textup{wts}}{\ge} \frac{\sin^2\left(\pi t\right)}{\pi^2}\sum_{n \in \mathbb{Z}\setminus\{\eta(x)\}}\left(\frac{F_\alpha(n)}{\left(x-n\right)^2}+\frac{F_\alpha'(n)}{x-n}\right).
    \end{split}
    \label{mi}
\end{equation}
We consider three cases, when $\eta(x)=0$, $\eta(x)=1$, and when $\eta(x)\ge 2$:

\subsection{Proving \texorpdfstring{\eqref{mi}}{Lg} for \texorpdfstring{$\eta(x)=0$}{Lg}}\label{etaxis0}
Here we want to show\begin{equation*}
    F_\alpha(x)-\frac{\sin^2\left(\pi x\right)}{\left(\pi x\right)^2}\underset{\textup{wts}}{\ge} \frac{\sin^2\left(\pi x\right)}{\pi^2}\sum_{\substack{n \in \mathbb{Z}\setminus\{0\}}}\left(\frac{F_\alpha(n)}{\left(x-n\right)^2}+\frac{F_\alpha'(n)}{x-n}\right),
\end{equation*} when $x \in \left[0,\frac1{2}\right]$. Using the identity \eqref{sumover1oversin^2} we multiply both sides of the above inequality by $\frac{\pi^2}{\sin^2\left(\pi x\right)}$ and use the fact the $F_\alpha'(n)$ is odd in $n$ to get the equivalent inequality
\begin{align*}
    F_\alpha(x)\sum_{n \in \mathbb{Z}}\frac1{\left(n-x\right)^2}-\frac1{x^2}\underset{\textup{wts}}{\ge} \sum_{n \in \mathbb{Z}\setminus\{0\}}\left(\frac{F_\alpha(n)}{\left(x-n\right)^2}+\frac{nF_\alpha'(n)}{x^2-n^2}\right),
\end{align*} which we write as
\begin{equation*}
    \frac{F_\alpha(x)-1}{x^2}+\sum_{n \in \mathbb{Z}\setminus\{0\}}\frac{F_\alpha(x)-F_\alpha(n)}{\left(x-n\right)^2}\underset{\textup{wts}}\ge \sum_{n \in \mathbb{Z}\setminus\{0\}}\frac{nF_\alpha'(n)}{x^2-n^2}.
\end{equation*} For $0\le x\le \frac1{2}$, $|n|\ge 1$, $\frac1{(n-x)^2}+\frac1{(n+x)^2}$ is increasing in $x$, and since $nF_\alpha'(n)\le 0$, for $0\le x\le \frac1{2}$, it is sufficient to show \begin{align}
    \frac{F_\alpha(x)-1}{x^2}+\sum_{n \in \mathbb{Z}\setminus\{0\}}\frac{F_\alpha\mathopen{}\left(\frac1{2}\right)-F_\alpha(n)}{n^2}\underset{\textup{wts}}{\ge} \sum_{n \in \mathbb{Z}\setminus\{0\}}\frac{nF_\alpha'(n)}{\frac1{4}-n^2}.
    \label{mazda}
\end{align} for $x \in \left[0,\frac1{2}\right]$ to prove \eqref{mi}. We can also observe $\frac{F_\alpha(x)-1}{x^2}$ is decreasing in $x$ for $x \in \left[0,\frac1{2}\right]$: Differentiating, \begin{align*}
    \partial_x\left(\frac{F_\alpha(x)-1}{x^2}\right)=\frac{F_\alpha'(x)x^2-2x\left(F_\alpha(x)-1\right)}{x^4},
\end{align*} so we have to show $F_\alpha'(x)x\le 2\left(F_\alpha(x)-1\right)$ for $0\le x\le \frac1{2}$ to show what we want. From \eqref{emu} $xF_\alpha'(x)=-\alpha F_\alpha(x)(1-F_\alpha(x))=\alpha F_\alpha(x)\left(F_\alpha(x)-1\right)$, so it is equivalent to show $F_\alpha(x)\ge \frac2{\alpha}$ for all $x \in \left[0,\frac1{2}\right]$, which follows from the fact that \begin{align*}
    F_\alpha\mathopen{}\left(\frac1{2}\right)\mathclose{}\ge 
    \frac1{1+\left(\frac1{2}\right)^\alpha}\ge \frac2{\alpha}
\end{align*} for $\alpha\ge 6$. This means $\frac{F_\alpha(x)-1}{x^2}$ is decreasing in $x$ for $x \in \left[0,\frac1{2}\right]$,  which means it is sufficient to show  \begin{align}
        4\left(F_\alpha\mathopen{}\left(\frac1{2}\right)\mathclose{}-1\right)+\sum_{n \in \mathbb{Z}\setminus\{0\}}\frac{F_\alpha\mathopen{}\left(\frac1{2}\right)-F_\alpha(n)}{n^2}\underset{\textup{wts}}{\ge} \sum_{n \in \mathbb{Z}\setminus\{0\}}\frac{nF_\alpha'(n)}{\frac1{4}-n^2}
        \label{melody}
    \end{align} to prove \eqref{mazda}. We verify this via interval arithmetic in Julia \cite{IntervalArithmetic.jl} when $6\le \alpha\le 10$, carried out in the following Github repository \cite{edwin2024}. For $\alpha\ge 12$, we prove it below.
\begin{proof}[Proof of \eqref{melody} for $\alpha\ge 12$.]
From Proposition \ref{thm:asympsalpha} and \eqref{emu}, we may deduce \begin{equation*}
    \begin{split}
        \sum_{n \in \mathbb{Z}\setminus\{0\}}\frac{nF_\alpha'(n)}{\frac1{4}-n^2}&=-\frac{8F_\alpha'(1)}{3}+2\sum_{n=2}^\infty\frac{nF_\alpha'(n)}{\frac1{4}-n^2}\le \frac{4\alpha}{3\alpha-6}+\sum_{n=2}^\infty \frac{2\alpha F_\alpha(n)}{4-\frac1{4}} \\
    &\le \frac{4\alpha}{3\alpha-6}+\sum_{n=2}^\infty\frac1{n^\alpha},
    \end{split}
\end{equation*} and from Lemma \ref{estimateforsumoverpowers},  we get \begin{align*}
    \sum_{n \in \mathbb{Z}\setminus\{0\}}\frac{nF_\alpha'(n)}{\frac1{4}-n^2}&\le \frac{4\alpha}{3\alpha-6}+\frac1{2^\alpha}\left(\frac{\alpha+1}{\alpha-1}\right).
\end{align*} The estimates in Proposition \ref{thm:asympsalpha} imply $F_\alpha\mathopen{}\left(\frac1{2}\right)-F_\alpha(1)\ge 0.94$  if $\alpha\ge 12$, and $F_\alpha\mathopen{}\left(\frac1{2}\right)-1\ge -0.01$, hence \begin{equation*}
    \begin{split}
        4\left(F_\alpha\mathopen{}\left(\frac1{2}\right)-1\right)+\sum_{n \in \mathbb{Z}\setminus\{0\}}\frac{F_\alpha\mathopen{}\left(\frac1{2}\right)\mathclose{}-F_\alpha(n)}{n^2}&\ge -0.04+\frac{0.94\pi^2}{3} \\
    &\ge \frac{4\alpha}{3\alpha-6}+\frac1{2^\alpha}\left(\frac{\alpha+1}{\alpha-1}\right),
    \end{split}
\end{equation*} for $\alpha\ge 12$, which proves \eqref{melody} for $\alpha\ge 12$.
\end{proof} 
This takes care of when $\eta(x)=0$. We now consider $\eta(x)\ge 1$.
Here we want to show with $x=1+t$, \begin{equation}
     \begin{split}
         &F_\alpha(1+t)-\frac{\sin^2\left(\pi t\right)}{\pi^2t^2}\left(F_\alpha\left(1\right)+tF_\alpha'\left(1\right)\right) \\
         &\underset{\textup{wts}}{\ge} \frac{\sin^2\left(\pi t\right)}{\pi^2}\sum_{n \in \mathbb{Z}\setminus\{1\}}\left(\frac{F_\alpha(n)}{\left(1+t-n\right)^2}+\frac{F_\alpha'(n)}{1+t-n}\right), 
     \end{split}
\end{equation}
 for all $t \in \left[-\frac1{2},\frac1{2}\right]$. We multiply both sides by $\frac{\pi^2}{\sin^2\left(\pi t\right)}$ using \eqref{sumover1oversin^2} to get the equivalent inequality \begin{align*}
     &F_\alpha\left(1+t\right)\sum_{n \in \mathbb{Z}}\frac1{\left(n-t\right)^2}-\frac{F_\alpha(1)+tF_\alpha'(1)}{t^2} \\
     &\underset{\textup{wts}}{\ge} \sum_{n \in \mathbb{Z}\setminus\{1\}}\left(\frac{F_\alpha(n)}{\left(1+t-n\right)^2}+\frac{F_\alpha'(n)}{1+t-n}\right),
 \end{align*} or \begin{equation*}
     \begin{split}
         &\frac{F_\alpha(1+t)-F_\alpha(1)-tF_\alpha'(1)}{t^2}+\sum_{n \in \mathbb{Z}\setminus\{0\}}\frac{F_\alpha(1+t)}{\left(n-t\right)^2} \\
     &\underset{\textup{wts}}{\ge} \sum_{n \in \mathbb{Z}\setminus\{1\}}\left(\frac{F_\alpha(n)}{\left(1+t-n\right)^2}+\frac{F_\alpha'(n)}{1+t-n}\right),
     \end{split}
 \end{equation*} 
 for all $t \in \left[-\frac1{2},\frac1{2}\right]$. We pull out the dominant terms in the right-hand side. Let \begin{equation}
  \mathcal{B}(\alpha,t)=\sum_{\substack{n \in \mathbb{Z} \\ |n|\ge 2}}\left(\frac{F_\alpha(n)}{(1+t-n)^2}+\frac{F_\alpha'(n)}{1+t-n}\right),
      \label{mathcalBalphatdef}
  \end{equation} so we want to show \begin{equation}
       \begin{split}
            &\frac{F_\alpha(1+t)-F_\alpha(1)-tF_\alpha'(1)}{t^2}+\sum_{n \in \mathbb{Z}\setminus\{0\}}\frac{F_\alpha(1+t)}{\left(n-t\right)^2} \\
        &\underset{\textup{wts}}{\ge} \frac1{(1+t)^2}+\frac{F_\alpha(1)}{(2+t)^2}-\frac{F_\alpha'(1)}{(2+t)}+\mathcal{B}(\alpha,t)
       \end{split}
        \label{vivalavida}
   \end{equation} for all $t\in \left[-\frac1{2},\frac1{2}\right]$. When $\alpha\le 1000$, we employ a numerical approach via interval arithmetic in Julia \cite{IntervalArithmetic.jl}, laid out in the following Github repository \cite{edwin2024}. For $\alpha\ge 1000$, we do this directly by hand.\newline
   
 \noindent
 We start by getting an upper bound for $\mathcal{B}(\alpha,t)$. From \eqref{emu}, $\left|F_\alpha'(n)\right|\le \alpha F_\alpha(n)$ for $n\ge 1$ and the bound on $s_\alpha$ in Proposition \ref{thm:asympsalpha} implies $F_\alpha(n)\le \alpha^{-1}n^{-\alpha}$, so if $|n|\ge 2$ and $|t|\le \frac1{2}$, then $|1+t-n|\ge |n|-\frac3{2}$ so \begin{align*}
     \left|\mathcal{B}(\alpha,t)\right|&\le 2\sum_{n=2}^\infty \left(\frac1{\left(n-\frac3{2}\right)^2}+\frac{\alpha}{n-\frac3{2}}\right)F_\alpha(n)\le 2\sum_{n=2}^\infty \left(\frac{4+2\alpha}{\alpha}\right)\frac1{n^\alpha} \\
     &\le \sum_{n=2}^\infty \frac5{n^\alpha}\le \frac5{2^\alpha}\left(\frac{\alpha+1}{\alpha-1}\right)\le \frac{10}{2^\alpha}
 \end{align*} for $\alpha\ge 1000$ and from Lemma \ref{estimateforsumoverpowers}, so it is sufficient to show for all $t\in \left[-\frac1{2},\frac1{2}\right]$, \begin{equation}
     \begin{split}
         &\frac{F_\alpha(1+t)-F_\alpha(1)-tF_\alpha'(1)}{t^2}+\sum_{n \in \mathbb{Z}\setminus\{0\}}\frac{F_\alpha(1+t)}{\left(n-t\right)^2} \\
         &\underset{\textup{wts}}{\ge} \frac1{(1+t)^2}+\frac{F_\alpha(1)}{(2+t)^2}-\frac{F_\alpha'(1)}{2+t}+\frac{10}{2^\alpha}.
     \end{split}
     \label{thingtoshow}
 \end{equation}
       We employ different approaches for $t\in \left[0,\frac1{2}\right]$, and $t \in \left[-\frac1{2},0\right)$:
       
\subsection{ Showing \texorpdfstring{\eqref{thingtoshow} when $t\in \left[0,\frac1{2}\right]$}{Lg}}
       Here we consider two cases.\newline

       \noindent
       \underline{\textbf{Case 1: Showing \eqref{thingtoshow} when $t \in \left[0,\frac1{2}\right]$ and $\alpha F_\alpha(1+t)\ge 10^{-2}$.} }\newline
       \noindent
     We use the mean value theorem on the left-hand side of \eqref{thingtoshow}, so it is then sufficient to show \begin{align}
      \frac1{2}\min_{s \in [0,t]}F_\alpha''(1+s)+\sum_{n \in \mathbb{Z}\setminus\{0\}}\frac{F_\alpha(1+t)}{(n-t)^2}\underset{\textup{wts}}{\ge} \frac1{(1+t)^2}+\frac{F_\alpha(1)}{(2+t)^2}-\frac{F_\alpha'(1)}{2+t}+\frac{10}{2^\alpha}. \label{alors}
  \end{align} From \eqref{Falphaduobleprime}, 
  \begin{align*}
       F_\alpha''(x)&=\alpha F_\alpha(x)\left(\frac{1-F_\alpha(x)}{x^2}\right)\left(\alpha\left(1-2F_\alpha(x)\right)+1\right),
  \end{align*} so for any $s \in [0,t] \subset \left[0,\frac1{2}\right]$,
  \begin{align*}
      F_\alpha''(1+s)&\ge \alpha^2F_\alpha(1+s)\left(\frac{1-F_\alpha(1)}{(1+s)^2}\right)\left(1-2F_\alpha(1)\right) \\
      &\ge \frac{4\alpha}{9}\cdot\alpha F_\alpha(1+t)\left(1-F_\alpha(1)\right)\left(1-2F_\alpha(1)\right) \\
      &\ge \frac{4\alpha \cdot 10^{-2}}{9}\left(1-F_\alpha(1)\right)\left(1-2F_\alpha(1)\right) \\
      &\ge \frac{4\alpha \cdot 10^{-2}}{9}\left(1-\frac1{2\alpha-4}\right)\left(1-\frac1{\alpha-2}\right),
  \end{align*} from the bounds of $F_\alpha(1)$ in Proposition \ref{thm:asympsalpha}. If $\alpha\ge 1000$, the right-hand side above is at least $2.21$. Going back to \eqref{alors}, it is then sufficient to show \begin{align*}
    2.21\underset{\textup{wts}}{\ge} \frac1{(1+t)^2}+\frac{F_\alpha(1)}{(2+t)^2}-\frac{F_\alpha'(1)}{2+t}+\frac{10}{2^\alpha}.
\end{align*} Using the bounds on $F_\alpha(1)$ and $F_\alpha'(1)$ in Proposition \ref{thm:asympsalpha}, we see that if $\alpha\ge 1000$, then  $F_\alpha(1)\le \frac1{1996}$, $-F_\alpha'(1)\le \frac{1000}{1996}$, so if $t \in \left[0,\frac1{2}\right]$, then \begin{align*}
    \frac1{(1+t)^2}+\frac{F_\alpha(1)}{(2+t)^2}-\frac{F_\alpha'(1)}{2+t}+\frac{10}{2^\alpha}&\le 1+\frac1{4\cdot 1996}+\frac1{2}\cdot \frac{1000}{1996}+\frac{10}{2^{1000}}\approx 1.251 \\
    &<2.21,
\end{align*} which takes care of \eqref{thingtoshow} for $\alpha\ge 1000$, in this case.\newline

\noindent
\underline{\textbf{Case 2: Showing \eqref{thingtoshow} when $t \in \left[0,\frac1{2}\right]$ and $\alpha F_\alpha(1+t)<10^{-2}$.}}\newline
\noindent
In this case, it is sufficient to show \begin{align*}
      -\frac{F_\alpha(1)}{t^2}-\frac{F_\alpha'(1)}{t}\underset{\textup{wts}}{\ge} \frac1{(1+t)^2}+\frac{F_\alpha(1)}{(2+t)^2}-\frac{F_\alpha'(1)}{2+t}+\frac{10}{2^\alpha},
  \end{align*}  or equivalently, \begin{align*}
      -\frac{2F_\alpha'(1)}{t(2+t)}\underset{\textup{wts}}{\ge} \frac1{(1+t)^2}+F_\alpha(1)\left(\frac1{t^2}+\frac1{(2+t)^2}\right)+\frac{10}{2^\alpha}.
  \end{align*} From the bounds on $F_\alpha(1)$ and $F_\alpha'(1)$ in Proposition \ref{thm:asympsalpha}, it is sufficient to show \begin{align}
      \frac{1-\frac1{2\alpha-4}}{t(2+t)}\underset{\textup{wts}}{\ge} \frac1{(1+t)^2}+\frac1{\alpha}\left(\frac1{t^2}+\frac1{(2+t)^2}\right)+\frac{10}{2^\alpha},
      \label{emote}
  \end{align} when $t \in \left[0,\frac1{2}\right]$ and $\alpha F_\alpha(1+t)<10^{-2}$. 
  The inequality $\alpha F_\alpha(1+t)<10^{-2}$ can be written explicitly as \begin{align*}
       \frac{\alpha}{1+s_\alpha^\alpha \left(1+t\right)^\alpha}<10^{-2}.
  \end{align*} From Proposition \ref{thm:asympsalpha}, $s_\alpha^\alpha \le 2\alpha$ if $\alpha\ge 12$, so the above inequality implies \begin{align*}
      \frac{\alpha}{1+2\alpha\left(1+t\right)^\alpha}<10^{-2} &\implies 2(1+t)^\alpha +\alpha^{-1}>100   \\
      &\implies 1+t>\left(50-\frac1{2\alpha}\right)^{\frac1{\alpha}}.
  \end{align*} Let $g:[0,0.01] \to [0,\infty)$ be given by \begin{align*}
      g(x)=\left(50-\frac{x}{2}\right)^x,
  \end{align*} so the above chain is saying that $1+t>g\left(\alpha^{-1}\right)$. We can note that $g$ is convex: By taking the second derivative, we get \begin{align*}
      g''(x)=g(x)\left(\left(\log\left(50-\frac{x}{2}\right)-\frac{x}{100-x}\right)^2-\left(\frac1{100-x}+\frac{100}{\left(x-100\right)^2}\right)\right),
  \end{align*}
  which is non-negative if $x \in [0.01]$. In that case, since $g$ is convex and $\alpha^{-1}<0.01$, the inequality $1+t>g\mathopen{}\left(\alpha^{-1}\right)$ implies \begin{align*}
      1+t>g(0)+g'(0)\alpha^{-1}.
  \end{align*} $g(0)=1$, and $g'(0)=\log (50)$, so we get $1+t>1+\alpha^{-1}\log(50)$, so $\alpha t>\log (50)$. In that case, we can write \begin{align*}
       &\frac{1-\frac1{2\alpha-4}}{t(2+t)}-\left(\frac1{(1+t)^2}+\frac1{\alpha}\left(\frac1{t^2}+\frac1{(2+t)^2}\right)+\frac{10}{2^\alpha}\right) \\
       &=w\left(\frac{1}{t}\left(\frac{1-\frac1{2\alpha-4}}{2+t}-\frac1{\alpha t}\right)-\frac1{(1+t)^2}\right) \\
       &+(1-w)\left(\frac{1-\frac1{2\alpha-4}}{t(2+t)}-\frac1{(1+t)^2}-\frac1{\alpha t^2}\right)-\frac1{\alpha(2+t)^2}-\frac{10}{2^\alpha}, \\
       &\ge w\left(\frac{1}{t}\left(\frac{1-\frac1{1996}}{2+t}-\frac1{\log 50}\right)-\frac1{(1+t)^2}\right) \\
       &+(1-w)\left(\frac{1-\frac1{1996}}{t(2+t)}-\frac1{(1+t)^2}-\frac1{\alpha t^2}\right)-\frac1{\alpha(2+t)^2}-\frac{10}{2^\alpha},
  \end{align*} for any $w \in [0,1]$. Now, if $t\in [0,0.1]$, we choose $w=1$, in which case we get \begin{align*}
       &\frac{1-\frac1{2\alpha-4}}{t(2+t)}-\left(\frac1{(1+t)^2}+\frac1{\alpha}\left(\frac1{t^2}+\frac1{(2+t)^2}\right)+\frac{10}{2^\alpha}\right) \\
       &\ge \frac1{t}\left(\frac{1-\frac1{1996}}{2+t}-\frac1{\log 50}\right)-\frac1{(1+t)^2}-\frac1{\alpha(2+t)^2}-\frac{10}{2^\alpha} \\
       &\ge 10\left(\frac{1-\frac1{1996}}{2.1}-\frac1{\log 50}\right)-1-\frac{1}{4000}-\frac{10}{2^{1000}}\ge 0,
  \end{align*} 
  and if $t \in \left[0.1,\frac1{2}\right]$, we choose $w=0$ which implies \begin{align*}
   &\frac{1-\frac1{2\alpha-4}}{t(2+t)}-\left(\frac1{(1+t)^2}+\frac1{\alpha}\left(\frac1{t^2}+\frac1{(2+t)^2}\right)+\frac{10}{2^\alpha}\right) \\
      &\ge \frac{1-\frac1{1996}}{t(2+t)}-\frac1{(1+t)^2}-\frac1{\alpha t^2}-\frac1{\alpha(2+t)^2}-\frac{10}{2^\alpha} \\
      &\ge  \frac{1-\frac1{1996}}{0.5(2.5)}-\frac1{(1.5)^2}-\frac1{1000\cdot 0.15^2}-\frac1{1000\cdot 2.15^2}-\frac{10}{2^{1000}}\ge 0.
  \end{align*} These two inequalities imply \begin{align*}
      \frac{1-\frac1{2\alpha-4}}{t(2+t)}-\left(\frac1{(1+t)^2}+\frac1{\alpha}\left(\frac1{t^2}+\frac1{(2+t)^2}\right)+\frac{10}{2^\alpha}\right)\ge 0
  \end{align*} when $t \in \left[0,\frac1{2}\right]$ and $\alpha F_\alpha(1+t)<10^{-2}$, as desired. 
This completes the proof for of \eqref{thingtoshow} for $t\in \left[0,\frac1{2}\right]$.

\subsection{Showing \texorpdfstring{\eqref{thingtoshow} when $t \in \left[-\frac1{2},0\right)$}{Lg}}
   Recall the inequality we want to show, \eqref{thingtoshow}: \begin{align*}
   &\frac{F_\alpha(1+t)-F_\alpha(1)-tF_\alpha'(1)}{t^2}+\sum_{n \in \mathbb{Z}\setminus\{0\}}\frac{F_\alpha(1+t)}{\left(n-t\right)^2} \\
   &\underset{\textup{wts}}{\ge} \frac1{(1+t)^2}+\frac{F_\alpha(1)}{(2+t)^2}-\frac{F_\alpha'(1)}{(2+t)}+\frac{10}{2^\alpha}.
   \end{align*} We consider two cases.\newline
   
   \noindent
   \underline{\textbf{Case 1: Showing \eqref{thingtoshow} when $t \in \left[-\frac1{2},0\right)$ and $ F_\alpha(1+t)< 0.48$.}} \newline
   \noindent
   By the mean value theorem, it is sufficient to show \begin{align}
      \frac1{2}\min_{s \in [t,0]}F_\alpha''(1+s)+\sum_{n \in \mathbb{Z}\setminus\{0\}}\frac{F_\alpha(1+t)}{(n-t)^2}\underset{\textup{wts}}{\ge} \frac1{(1+t)^2}+\frac{F_\alpha(1)}{(2+t)^2}-\frac{F_\alpha'(1)}{2+t}+\frac{10}{2^\alpha},
      \label{moone}
  \end{align} when $t \in \left[-\frac1{2},0\right)$ and $ F_\alpha(1+t)< 0.48$. From \eqref{Falphaduobleprime}, \begin{align*}
      F_\alpha''(1+s)&=\alpha F_\alpha(1+s)\left(\frac{1-F_\alpha(1+s)}{(1+s)^2}\right)\left(\alpha\left(1-2F_\alpha(1+s)\right)+1\right),  \end{align*}  and if $ F_\alpha(1+t)< 0.48$ and $1>1+s>1+t$, then $ F_\alpha(1)\le F_\alpha(1+s)<0.48$, so \begin{align*}
      F_\alpha''(1+s)&\ge \alpha^2  F_\alpha(1)(1-0.48)\left(1-2\cdot 0.48\right) \\
      &\ge \frac{\alpha}{2}\left(1-0.48\right)\left(1-2\cdot 0.48\right)=0.0104\alpha,
  \end{align*} since $F_\alpha(1)\ge \frac{1}{2\alpha}$ from Proposition \ref{thm:asympsalpha}. So to show \eqref{moone}, it is sufficient to show \begin{align}
      0.0052\alpha \underset{\textup{wts}}{\ge} \frac1{(1+t)^2}+\frac{F_\alpha(1)}{(2+t)^2}-\frac{F_\alpha'(1)}{2+t}+\frac{10}{2^\alpha}.
      \label{everybodyathtebargettingtipsy}
  \end{align} We know $F_\alpha(1)\le \frac1{2\alpha-4}$ and $-F_\alpha'(1)\le \frac{\alpha}{2\alpha-4}$, from Proposition \ref{thm:asympsalpha}, so if $\alpha\ge 1000$ and $t \in \left[-\frac1{2},0\right)$, then \begin{align*}
      \frac1{(1+t)^2}+\frac{F_\alpha(1)}{(2+t)^2}-\frac{F_\alpha'(1)}{2+t}+\frac{10}{2^\alpha}\le 4+\frac{4}{9}\cdot \frac1{1996}+\frac2{3}\cdot \frac{1000}{1996}+\frac{10}{2^{1000}}<5.
  \end{align*} If $\alpha\ge 1000$, then $0.0052\alpha\ge 5.2$, so this shows \eqref{everybodyathtebargettingtipsy}, and takes care of Case 1.\newline
  
  \noindent
  \underline{\textbf{Case 2: Showing \eqref{thingtoshow} when $t \in \left[-\frac1{2},0\right)$ and $ F_\alpha(1+t)\ge 0.48$.}}\newline
  \noindent 
  From Proposition \ref{thm:asympsalpha} we may deduce $F_\alpha(1)\le \frac1{\alpha}$, $F_\alpha'(1)\le \frac{\alpha}{2\alpha-4}$, and since we are considering $t \in \left[-\frac1{2},0\right)$, it is sufficient to show \begin{equation}
  \begin{split}
      &\frac{F_\alpha(1+t)-F_\alpha(1)-tF_\alpha'(1)}{t^2}+\sum_{n \in \mathbb{Z}\setminus\{0\}}\frac{F_\alpha(1+t)}{\left(n-t\right)^2} \\
      &\underset{\textup{wts}}{\ge} \frac1{(1+t)^2}+\frac{4}{9\alpha}+\frac{\frac{\alpha}{2\alpha-4}}{2+t}+\frac{10}{2^\alpha}.
  \end{split}
       \label{weakenversionof4.50}
  \end{equation}  We start by showing that $ F_\alpha(1+t)-F_\alpha(1)-tF_\alpha'(1)\ge 0$ if $-\frac1{2}\le t\le 0$.
  \begin{proof}
      First suppose $t$ is such that $F_\alpha(1+t)<\frac1{2}$. Then by the mean value theorem, \begin{align*}
     \frac{F_\alpha(1+t)-F_\alpha(1)-tF_\alpha'(1)}{t^2}=\frac1{2}F''_\alpha(1+s)
\end{align*} for some $s \in [t,0]$, and if $s\ge t$, then $F_\alpha(1+s)\le F_\alpha(1+t)<\frac1{2}$, so from \eqref{Falphaduobleprime}, $F_\alpha''(1+s)\ge 0$, and so \begin{align*}
    F_\alpha(1+t)-F_\alpha(1)-tF_\alpha'(1)\ge 0.
\end{align*} If on the other hand $t \in \left[-\frac1{2},0\right)$ is such that $F_\alpha(1+t)\ge \frac1{2}$, then from the bounds of $F_\alpha(1)$, $F_\alpha'(1)$ in Proposition \ref{thm:asympsalpha}, \begin{align*}
    F_\alpha(1+t)-F_\alpha(1)-tF_\alpha'(1)\ge \frac1{2}-\frac1{2\alpha-4}-\frac{\alpha}{4\alpha-8}\ge 0,
\end{align*} since $|t|\le \frac1{2}$ and $\alpha\ge 1000$, so $F_\alpha(1+t)-F_\alpha(1)-tF_\alpha'(1)\ge 0$ for $t \in \left[-\frac1{2},0\right]$, $\alpha\ge 1000$
  \end{proof}
    This implies it is sufficient to show
\begin{align}
    \sum_{n \in \mathbb{Z}\setminus\{0\}}\frac{F_\alpha(1+t)}{\left(n-t\right)^2}\underset{\textup{wts}}{\ge} \frac1{(1+t)^2}+\frac{4}{9\alpha}+\frac{\frac{\alpha}{2\alpha-4}}{2+t}+\frac{10}{2^\alpha}
    \label{furelise}
\end{align} for $t \in \left[-\frac1{2},0\right]$ such that $F_\alpha(1+t)\ge 0.48$, to show \eqref{weakenversionof4.50}.\newline

\noindent
If $t\ge -0.1$, we use the fact that $\sum_{n \in \mathbb{Z}\setminus \{0\}}\frac1{(n-t)^2}$ is increasing in $|t|$, in which case it is sufficient to show \begin{align*}
    0.48\sum_{n \in \mathbb{Z}\setminus\{0\}}\frac1{n^2}=\frac{0.48\pi^2}{3}\ge \frac1{0.9^2}+\frac4{9\alpha}+\frac{\frac{\alpha}{2\alpha-4}}{1.9}+\frac{10}{2^\alpha},
\end{align*} which holds if $\alpha\ge 1000$. If $t\le -0.1$, then $F_\alpha(1+t)\ge F_\alpha(0.9)$, so it suffices to show \begin{align*}
    \sum_{n \in \mathbb{Z}\setminus \{0\}}\frac{F_\alpha(0.9)}{(n-t)^2}\ge \frac1{(1+t)^2}+\frac{4}{9\alpha}+\frac{\frac{\alpha}{2\alpha-4}}{2+t}+\frac{10}{2^\alpha}.
\end{align*} From the bounds of $F_\alpha(1)$, $F_\alpha'(1)$ in Proposition \ref{thm:asympsalpha}, we may deduce $F_\alpha(0.9)>0.999$ for $\alpha\ge 1000$. One can check that \begin{align*}
    \frac{0.999}{\left(1+t\right)^{2}}+\frac{0.999}{\left(1-t\right)^{2}}\ge \frac{1}{(1+t)^{2}}+\frac{4}{9\alpha}+\frac{\frac{\alpha}{2\alpha-4}}{2+t}+\frac{10}{2^{\alpha}}
\end{align*}  for $-\frac1{2}\le t\le -0.1$, which implies the desired inequality.
This completes the proof of \eqref{vivalavida}, taking care of the $\eta(x)=1$ case.\newline

It remains to show \eqref{mi} for $\eta(x)\ge 2$ to complete the proof of $\psi_\alpha(x)\le F_\alpha(x)$, and hence Theorem \ref{measure}.

\subsection{Proving \texorpdfstring{\eqref{mi}}{Lg} for \texorpdfstring{$\eta(x)\ge 2$}{Lg}:}\label{etaxis2}
We go back to the inequality we want to show, \eqref{mi}: \begin{equation}
    \begin{split}
        &F_\alpha(x)-\frac{\sin^2\left(\pi t\right)}{\pi^2t^2}\left(F_\alpha\left(\eta(x)\right)+tF_\alpha'\left(\eta(x)\right)\right) \\
    &\underset{\textup{wts}}{\ge} \frac{\sin^2\left(\pi t\right)}{\pi^2}\sum_{n \in \mathbb{Z}\setminus\{\eta(x)\}}\left(\frac{F_\alpha(n)}{\left(x-n\right)^2}+\frac{F_\alpha'(n)}{x-n}\right),
    \end{split}
    \label{hozierr}
\end{equation} where $x=\eta(x)+t$. We start by observing that the left-hand side above is non-negative. This is clearly true if $F_\alpha\left(\eta(x)\right)+tF_\alpha'\left(\eta(x)\right)\le 0$, so we consider the case where $F_\alpha\left(\eta(x)\right)+tF_\alpha'\left(\eta(x)\right)>0$, in which case we can right the left-hand side of \eqref{hozierr} as \begin{align*}
    &F_\alpha\mathopen{}\left(\eta(x)+t\right)\mathclose{}-F_\alpha\left(\eta(x)\right)-tF_\alpha'\mathopen{}\left(\eta(x)\right)\mathclose{}\\
    &+\mathopen{}\left(1-\frac{\sin^2\left(\pi t\right)}{\pi^2t^2}\right)\mathclose{}\mathopen{}\left(F_\alpha\left(\eta(x)\right)+tF_\alpha'\left(\eta(x)\right)\right)\mathclose{},
\end{align*} so it is sufficient to show \begin{equation*}
    F_\alpha\mathopen{}\left(\eta(x)+t\right)\mathclose{}-F_\alpha\left(\eta(x)\right)-tF_\alpha'\left(\eta(x)\right)\underset{\textup{wts}}{\ge} 0,
\end{equation*} when $F_\alpha\left(\eta(x)\right)+tF_\alpha'\left(\eta(x)\right)>0$. By the mean value theorem, we know  \begin{equation*}
    F_\alpha\left(\eta(x)+t\right)-F_\alpha\left(\eta(x)\right)-tF_\alpha'\left(\eta(x)\right)\ge \frac1{2}\min_{|t|\le \frac1{2}}F_\alpha''\left(\eta(x)+t\right).
\end{equation*} Now if $\eta(x)\ge 2$ and $|t|\le \frac1{2}$, then $\eta(x)+t\ge 2-\frac1{2}=\frac3{2}$ and so $F_\alpha\left(\eta(x)+t\right)\le F_\alpha\mathopen{}\left(\frac3{2}\right)\mathclose{}<\frac1{2}$. From the formula for $F_\alpha''$ in \eqref{Falphaduobleprime}, this means $F_\alpha''\left(\eta(x)+t\right)\ge 0$, and so \begin{equation*}
    F_\alpha\left(\eta(x)+t\right)-F_\alpha\left(\eta(x)\right)-tF_\alpha'\left(\eta(x)\right)\ge 0,
\end{equation*} as desired. This shows the left-hand side of \eqref{hozierr} is non-negative, so to show \eqref{hozierr}, it is sufficient to show \begin{align}
 \sum_{n \in \mathbb{Z}\setminus\{\eta(x)\}}\left(\frac{F_\alpha(n)}{\left(x-n\right)^2}+\frac{F_\alpha'(n)}{x-n}\right) \underset{\textup{wts}}{\le} 0.
    \label{madisonbeer}
\end{align} 
When $6\le \alpha\le 14$ and $1.5\le x\le 10$, we verify this using interval arithmetic in Julia \cite{IntervalArithmetic.jl}, available in the following Github repository \cite{edwin2024}. For the complementary region, we prove that by hand:
\begin{proof}[Proof of \eqref{madisonbeer} when $6\le \alpha\le 14$ and $x\ge 10$, or $\alpha\ge 16$ and $x\ge 1.5$]
    We extract the $n=-\eta(x)$ term in \eqref{madisonbeer}, so it is equivalent to show \begin{equation*}
    \sum_{\substack{n \in \mathbb{Z} \\ |n| \neq \eta(x)}}\left(\frac{F_\alpha(n)}{\left(x-n\right)^2}+\frac{F_\alpha'(n)}{x-n}\right)\underset{\textup{wts}}{\le} -\left(\frac{F_\alpha\left(\eta(x)\right)}{\left(x+\eta(x)\right)^2}-\frac{F_\alpha'\left(\eta(x)\right)}{x+\eta(x)}\right).
\end{equation*} From \eqref{ipadair},  we may deduce \begin{align*}
    &x^4 \sum_{\substack{n \in \mathbb{Z} \\ |n| \neq \eta(x)}}\left(\frac{F_\alpha(n)}{\left(x-n\right)^2}+\frac{F_\alpha'(n)}{x-n}\right) \\
    &=-2x^2\left(F_\alpha(\eta(x))+\eta(x)F_\alpha'(\eta(x))\right)+\sum_{\substack{n \in \mathbb{Z} \\ |n|\neq \eta(x)}}\left(3n^2F_\alpha(n)+n^3F_\alpha'(n)\right) \\
    &+\sum_{\substack{n \in \mathbb{Z} \\ |n|\neq \eta(x)}}\frac{2n^6F_\alpha(n)}{\left(x^2-n^2\right)}+\sum_{\substack{n \in \mathbb{Z} \\ |n|\neq \eta(x)}}\frac{5n^4F_\alpha(n)+n^5F_\alpha'(n)}{x^2-n^2},
\end{align*} so it is sufficient to show \begin{equation}
\begin{split}
    &\sum_{\substack{n \in \mathbb{Z}}}\left(3n^2F_\alpha(n)+n^3F_\alpha'(n)\right)+\sum_{\substack{n \in \mathbb{Z} \\ |n|\neq \eta(x)}}\left(\frac{2n^6F_\alpha(n)}{\left(x^2-n^2\right)}+\frac{5n^4F_\alpha(n)+n^5F_\alpha'(n)}{x^2-n^2}\right) \\
    &\underset{\textup{wts}}{\le}  2x^2\left(F_\alpha(\eta(x))+\eta(x)F_\alpha'(\eta(x))\right)+6\eta(x)^2F_\alpha(\eta(x))+2\eta(x)^3F_\alpha'(\eta(x)) \\
    &-x^4\left(\frac{F_\alpha\left(\eta(x)\right)}{\left(x+\eta(x)\right)^2}-\frac{F_\alpha'\left(\eta(x)\right)}{x+\eta(x)}\right).
\end{split}
    \label{worn}
\end{equation}
Note the inequality in \eqref{obrien} implies \begin{align*}
    \sum_{\substack{n \in \mathbb{Z} \\ |n|\neq \eta(x)}}\frac{2n^6F_\alpha(n)}{\left(x^2-n^2\right)^2}\le \frac2{s_\alpha^\alpha}\sum_{\substack{n \in \mathbb{Z} \\ |n|\neq \eta(x)}}\frac1{\left(x^2-n^2\right)^2}\le \frac{16}{s_\alpha^\alpha x^2},
\end{align*} and if $|n|\neq \eta(x)$, then $\left|x^2-n^2\right|\ge \frac{x}{2}$, hence \begin{equation*}
    \begin{split}
        \sum_{\substack{n \in \mathbb{Z} \\ |n|\neq k}}\frac{5n^4F_\alpha(n)+n^5F_\alpha'(n)}{x^2-n^2}&\le \frac{10F_\alpha(1)+2F_\alpha'(1)}{x^2-1} \\
&+\frac2{x}\sum_{n=2}^\infty\left|10n^4F_\alpha(n)+2n^5F_\alpha'(n)\right|.
    \end{split}
\end{equation*} Combining these two inequalities, we get \begin{equation}
\begin{split}
    & \sum_{\substack{n \in \mathbb{Z}}}\left(3n^2F_\alpha(n)+n^3F_\alpha'(n)\right)+\sum_{\substack{n \in \mathbb{Z} \\ |n|\neq \eta(x)}}\left(\frac{2n^6F_\alpha(n)}{\left(x^2-n^2\right)}+\frac{5n^4F_\alpha(n)+n^5F_\alpha'(n)}{x^2-n^2}\right)  \\
   &\le \sum_{\substack{n \in \mathbb{Z}}}\left(3n^2F_\alpha(n)+n^3F_\alpha'(n)\right)+\frac{16}{s_\alpha^\alpha x^2}+\frac{10F_\alpha(1)+2F_\alpha'(1)}{x^2-1}\\
   &+\frac2{x}\sum_{n=2}^\infty\left|10n^4F_\alpha(n)+2n^5F_\alpha'(n)\right|.
\end{split}
    \label{durk}
\end{equation}
 On the other hand, if $\eta(x)\ge 2$, \begin{align*}
 &2x^2\left(F_\alpha(\eta(x))+\eta(x)F_\alpha'(\eta(x))\right)+6\eta(x)^2F_\alpha(\eta(x))+2\eta(x)^3F_\alpha'(\eta(x)) \\
 &-x^4\left(\frac{F_\alpha\left(\eta(x)\right)}{\left(x+\eta(x)\right)^2}-\frac{F_\alpha'\left(\eta(x)\right)}{x+\eta(x)}\right) \\
 &\ge 2x^2\eta(x)F_\alpha'(\eta(x))+2\eta(x)^3F_\alpha'(x)-x^2F_\alpha(\eta(x))+x^3F_\alpha'(\eta(x)) \\
 &\ge 8\eta(x)^3F_\alpha'(\eta(x))-2\eta(x)^2F_\alpha(x) .
\end{align*} From \eqref{emu}, $xF_\alpha'(x)\ge -\alpha F_\alpha(x)$, so this implies \begin{align*}
    &2x^2\left(F_\alpha(\eta(x))+\eta(x)F_\alpha'(\eta(x))\right)+6\eta(x)^2F_\alpha(\eta(x))+2\eta(x)^3F_\alpha'(\eta(x)) \\
    &-x^4\left(\frac{F_\alpha\left(\eta(x)\right)}{\left(x+\eta(x)\right)^2}-\frac{F_\alpha'\left(\eta(x)\right)}{x+\eta(x)}\right) \\
    & \ge -(2+8\alpha)\eta(x)^2F_\alpha(\eta(x))\ge -\frac{8\alpha+2}{\eta(x)^{\alpha-2}}.
\end{align*} Putting this and \eqref{durk} together, we see that to prove \eqref{worn}, it is sufficient to show \begin{equation}
   \begin{split}
        &\sum_{\substack{n \in \mathbb{Z}}}\left(3n^2F_\alpha(n)+n^3F_\alpha'(n)\right)+\frac{16}{s_\alpha^\alpha x^2}+\frac{10F_\alpha(1)+2F_\alpha'(1)}{x^2-1} \\
    &+\frac2{x}\sum_{n=2}^\infty\left|10n^4F_\alpha(n)+2n^5F_\alpha'(n)\right|\underset{\textup{wts}}{\le} -\frac{8\alpha+2}{\eta(x)^{\alpha-2}}.
   \end{split}
    \label{allthestars}
\end{equation} 
If $6\le \alpha\le 14$ and $x\ge 10$, we can see this is true, since \begin{align*}
     &\sum_{\substack{n \in \mathbb{Z}}}\left(3n^2F_\alpha(n)+n^3F_\alpha'(n)\right)+\frac{16}{100s_\alpha^\alpha}+\frac{10F_\alpha(1)+2F_\alpha'(1)}{99} \\
     &+\sum_{n=2}^\infty\frac{\left|10n^4F_\alpha(n)+2n^5F_\alpha'(n)\right|}{5}+\frac{8\alpha+2}{10^{\alpha-2}}\le 0.
\end{align*} This is verified via interval arithmetic in Julia \cite{IntervalArithmetic.jl}, in the following Github repository \cite{edwin2024}.
If $\alpha \ge 16$, we can prove \eqref{allthestars} for all $x\ge 1.5$: It is enough to check that \begin{align*}
     &\sum_{\substack{n \in \mathbb{Z}}}\left(3n^2F_\alpha(n)+n^3F_\alpha'(n)\right)+\frac{16}{2.25s_\alpha^\alpha }+\frac{10F_\alpha(1)+2F_\alpha'(1)}{1.25} \\
     &+\frac2{1.5}\sum_{n=2}^\infty\left|10n^4F_\alpha(n)+2n^5F_\alpha'(n)\right| \le -\frac{8\alpha+2}{2^{\alpha-2}},
\end{align*} done below. If $\alpha\ge 16$, then \begin{align*}
    \sum_{\substack{n \in \mathbb{Z} \\ |n|\ge 2}}n^4F_\alpha(n)\le \frac1{s_\alpha^\alpha}\sum_{n=2}^\infty \frac2{n^{\alpha-4}}\le \frac1{s_\alpha^\alpha}\cdot \frac2{2^{\alpha-4}}\left(\frac{\alpha-3}{\alpha-5}\right)
\end{align*} from Lemma \ref{estimateforsumoverpowers}, and from  \eqref{emu}, $-xF_\alpha'(x)\le \alpha F_\alpha(x)$, so this implies \begin{align*}
     \sum_{n \in \mathbb{Z}}\left(3n^2F_\alpha(n)+n^3F_\alpha'(n)\right)&\le 6F_\alpha(1)+2F_\alpha'(1)+\frac{6}{s_\alpha^\alpha 2^{\alpha-4}}\left(\frac{\alpha-3}{\alpha-5}\right) \\
     &\le -1+\frac7{2\alpha-4}+\frac1{2^{\alpha-4}}\cdot \frac{6(\alpha-3)}{(2\alpha-5)(\alpha-5)} \\
     &\le -1+\frac7{2\alpha-4}+\frac1{2^{\alpha-4}}
\end{align*} if $\alpha\ge 16$, from the bounds in Proposition \ref{thm:asympsalpha}. Similarly, \begin{align*}
      \sum_{n =2}^\infty\left|10n^4F_\alpha(n)+2n^5F_\alpha'(n)\right|&\le \sum_{n=2}^\infty \left(10+2\alpha\right)n^4F_\alpha(n)\le  \frac{2\alpha+10}{s_\alpha^\alpha 2^{\alpha-4}}\left(\frac{\alpha-3}{\alpha-5}\right) \\
      &\le \frac{2\alpha+10}{(2\alpha-5)2^{\alpha-4}}\left(\frac{\alpha-3}{\alpha-5}\right)\le \frac2{2^{\alpha-4}},
\end{align*} and \begin{align*}
    10F_\alpha(1)+2F_\alpha'(1)\le \frac{11}{2\alpha-4}-1,
\end{align*} from the bounds in Proposition \ref{thm:asympsalpha}, if $\alpha\ge 16$. Consequently, if $\eta(x)\ge 2$ and $\alpha\ge 16$, then \begin{align*}
    &\sum_{\substack{n \in \mathbb{Z}}}\left(3n^2F_\alpha(n)+n^3F_\alpha'(n)\right)+\frac{16}{2.25s_\alpha^\alpha }+\frac{10F_\alpha(1)+2F_\alpha'(1)}{1.25} \\
    &+\frac2{1.5}\sum_{n=2}^\infty\left|10n^4F_\alpha(n)+2n^5F_\alpha'(n)\right| \\
    &\le -1+\frac7{2\alpha-4}+\frac1{2^{\alpha-4}}+\frac{\frac{11}{2\alpha-4}-1}{1.25}+\frac{16}{2.25s_\alpha^\alpha}+\frac4{1.5\cdot 2^{\alpha-4}}\le -\frac{8\alpha+2}{2^{\alpha-2}},
\end{align*} as desired.
\end{proof}

\section*{Acknowledgements}
The authors would like to thank Professor Henry Cohn of the MIT Department of Mathematics for his advice, mentorship and comments in the course of writing this paper.

\end{document}